\documentclass[twocolumn,12pt]{autart}

\usepackage[dvips]{epsfig}
\usepackage{amsmath}
\usepackage{ntheorem}

\usepackage{amssymb}
\usepackage{mathtools}
\usepackage{amsfonts}
\usepackage{color}
\usepackage{subfigure}
\usepackage{algorithm}
\usepackage{uri}
\usepackage{bbm}

\newtheorem*{proof}{Proof}
\newtheorem{lemma}{Lemma}
\newtheorem{theorem}{Theorem}

\newtheorem{remark}{Remark}

\usepackage{enumitem}
\renewcommand{\theenumi}{\arabic{enumi}}

\renewcommand{\theenumii}{\arabic{enumii}}

\renewcommand{\theenumiii}{\arabic{enumiii}}

\newcommand{\aseq}{\displaystyle{\mathop{=}^{\cdot}}}

\newcommand{\E}{{\mathbb E}}
\newcommand{\R}{{\mathbb R}}
\newcommand{\Z}{{\mathbb Z}}

\newcommand{\rev}[1]{{\textcolor{black}{#1}}}

\newcommand{\AC}[1]{{\textcolor{black}{#1}}}

\begin{document}
	
	\begin{frontmatter}
		
		\title{\rev{Data-driven predictive control in a stochastic setting:\\a unified framework}\thanksref{footnoteinfo}}
		
		\thanks[footnoteinfo]{This project was partially supported by the Italian Ministry of University and Research under the PRIN’17 project \textquotedblleft Data-driven learning of constrained control systems\textquotedblright, contract no. 2017J89ARP. Corresponding author: Simone Formentin (e-mail: simone.formentin@polimi.it).}
		
		\author[POLIMI]{Valentina Breschi},
		\author[UNIPD]{Alessandro Chiuso},
		\author[POLIMI]{Simone Formentin}
		
		\address[POLIMI]{Dipartimento di Elettronica, Informazione e Bioingegneria, Politecnico di Milano, P.za L. Da Vinci, 32, 20133 Milano, Italy.}
		\address[UNIPD]{Department of Information Engineering, University of Padova, Via Gradenigo 6/b, 35131 Padova, Italy.}

		
		\begin{abstract}
			Data-driven predictive control (DDPC) has been recently proposed as an effective alternative to traditional model-predictive control (MPC) for its unique features of being time-efficient and unbiased with respect to the oracle solution. Nonetheless, it has also been observed 
			that 
			noise 
			may strongly jeopardize the final closed-loop performance, since it affects both the data-based system representation and the control update computed from the online measurements. Recent studies have shown that regularization is potentially a successful tool to counteract the effect of noise. \rev{At the same time, regularization requires the tuning of a set of penalty terms, whose choice might be practically difficult without closed-loop experiments.} \rev{In this paper, by means of subspace identification tools, we pursue a three-fold goal: $(i)$ we set up a unified framework for the existing regularized data-driven predictive control schemes {for stochastic systems}; $(ii)$ we introduce $\gamma$-DDPC, an efficient two-stage scheme that {splits the optimization problem in two parts:  fitting the initial conditions and optimizing the future performance, while guaranteeing constraint satisfaction}; $(iii)$ we \rev{discuss} the role of regularization for data-driven predictive control, providing new insight on }\emph{when} and \emph{how} it should be applied. A benchmark numerical case study finally illustrates \rev{the performance of $\gamma$-DDPC, showing how controller design can be simplified in terms of tuning effort and computational complexity when benefiting from the insights coming from the subspace identification realm.} 
		\end{abstract}

	\end{frontmatter}
	
	\section{Introduction}
\rev{Data}-driven control \rev{(DDC)} refers to the science of learning feedback controllers from data, without first undertaking a full modeling study of the plant to control \cite{formentin2014comparison}. Such a direct mapping of data onto the control action is indeed advisable in real-world problems, as modeling usually takes about 75\% of the time devoted to a control project \cite{gevers2005identification}, and accurate modeling for control requires significant time and several (costly) technical expertises, e.g., in the process domain and in the statistical tools for system identification. {Additionally, accurate modeling may go well beyond what is strictly necessary for control purposes only, since often times rather limited knowledge of the system dynamics may be required to achieve the desired control objectives \cite{Hjalmasson-05}.}  
Early attempts in \rev{the} direction \rev{of DDC} date back to 1942, with the first studies by Ziegler and Nichols about PID auto-tuning \cite{ziegler1942optimum}. More sophisticated, optimization based, approaches have been derived since then for fixed-order controller tuning, leading to a portfolio of techniques suitable for different problem formulations, see, e.g., \cite{formentin2019deterministic,breschi2020direct,formentin2012data,karimi2017data}. However, it is only recently that such a paradigm shift in control design could be extended to more complex control architectures, \rev{thanks to }
the availability of large datasets and unparalleled computing power. 

\rev{In this context lays the uprising interest in data-driven predictive control (DDPC) solutions, that combine the capability of constraint handling of MPC with the flexibility of a data-driven, nonparametric predictor of the system under control. By relying on the} so-called ``fundamental lemma'' \cite{willems2005note} \rev{(or variations of the latter)}, \rev{most of existing DDPC techniques replace} model equations with suitable  data-based constraints\footnote{Such constraints \rev{are} an implicit, nonparametric, mapping of the input/output relationships. According to this interpretation, some researchers legitimately prefer to denote the strategies described herein as ``indirect''. For this reason, we will simply talk about \textit{data-driven predictive control} from now on.} \rev{(see e.g., \cite{berberich2020data,coulson2019data})}. \rev{The transition to this} data-based framework may lead to different performance than traditional model-based MPC, \rev{because of its} unique features. For instance, the sub-optimality gap measuring the control performance with respect to the optimal model-based \rev{solution} (namely, \rev{that obtained} using the real model of the system) vanishes with the size of the dataset. Moreover, model-free predictive control may indirectly address the bias/variance trade-off in a more efficient manner\rev{. Indeed, } it will not  incur in the asymptotic  bias induced by inaccurate modeling when complexity constraints are imposed on the model structure, as discussed in \cite{krishnan2021direct}. 

The transition from a model-based to a data-driven framework is well established in case of purely deterministic systems, whereas many of the attempts made to counteract the effect of noise in the presence of stochastic disturbances lead to approximations that may deteriorate the closed-loop performance. {For instance,} in \cite{berberich2020data}, the authors \rev{propose a regularized DDPC scheme, with guarantees of practical exponential stability in closed loop}  in the presence of \textit{bounded} additive output noise. The key ingredients \rev{to achieve this result} are two: $(i)$ some bounded slack variables to account for the noisy data used for prediction, and $(ii)$ some suitable regularization terms. In \cite{alanwar2021robust}, a slightly different scheme is used, which computes the data-driven reachable set based on a matrix zonotope recursion starting from the measured output. For this scheme, the authors show they can guarantee robust constraint satisfaction, again in case of \textit{bounded} process and measurement noises. The case of stochastic (white) measurement noise is addressed in \cite{yin2021maximum}, where a maximum likelihood framework is proposed to estimate the data-based constraints aimed to replace the model equations in the MPC formulation. {the resulting scheme is an iterative \textit{two-stage} approach, where at each iteration first a model encoded by a data matrix constraint \textit{must} be identified and then the online predictive control is computed}. An approach to handle stochastic noise in the direct framework proposed in \cite{coulson2019data} and \cite{berberich2020data} is given in the recent paper \cite{dorfler2022bridging}. In this contribution, the authors {exploit regularization as the \textit{key tool} to handle the presence of noise in the output measurements, and empirically discuss the performance of different regularization schemes. \rev{Approaches for DDPC with regularization are also shown to be distributionally robust in \cite{Coulson2019}. One of the regularized schemes proposed in \cite{dorfler2022bridging} is then connected with \emph{Subspace Predictive Control} (SPC) \cite{FAVOREEL1999} by \cite{Felix2021}, where the introduction of additional slacks is further propose to cope with noisy online data. {Kalman filter approaches have finally been suggested in \cite{Alpago2020} to filter out the effect of noise in the context of DeePC approaches.}}

In this paper, we consider a stochastic setting where both measurement and process noise are considered. Within this framework, our contribution is three-fold. 
\begin{enumerate}
	\item[C1.] By revising foundational results in subspace identification, we show that the seminal regularized DDPC schemes in \cite{Felix2021,berberich2020data,dorfler2022bridging} can all be recast into a \emph{unified framework}, stemming from the constrained counterpart of the SPC scheme originally proposed in \cite{FAVOREEL1999}. 
	\item[C2.] Based on this unified framework, we discuss how the choice of \emph{key hyperparameters} in \cite{Felix2021,berberich2020data,dorfler2022bridging} can be guided by known results in subspace identification. These insights potentially allow the final user to select the regularization parameters in those schemes with \emph{less} closed-loop tests, while possibly allow one to \emph{avoid} such experiments if the available dataset is large.
	\item[C3.] We show that the {parameterization of the predictor and the control input exploited to solve the DDPC problem} can be  decomposed in three terms with specific roles. This decomposition allows us to \emph{split} the DDPC problem into two sub-problems of smaller dimensions, respectively devoted to: $(i)$ fit the initial conditions embedded in the input/output data streams collected online; $(ii)$ optimize performance in prediction, while avoiding constraint violations. The introduction of this two-stage scheme, which we call \textit{$\gamma$-DDPC} from now on, allows for a reduction in the computational complexity of the overall DDPC formulation, while providing the final user with a more transparent overview of the main players of the control scheme.
\end{enumerate}
By means of a benchmark numerical example, we show the performance of $\gamma$-DDPC and we validate the insights gained from subspace identification about the role of regularization, showing how the latter can be actively exploited to \textit{avoid} (or at least reduce the number of) the closed loop experiments} needed to tune the regularization weights through cross-validation.

The remainder of the paper is \rev{organized} as follows. In Section \ref{Sec:goal}, we formally define the control problem of interest and its data-driven counterpart. Section \ref{sec:SID} \rev{reviews in details subspace identification concepts to} give a deeper insight into the employed system description\rev{, ultimately leading to the constrained SPC formulation at the core of the unified {framework for} regularized DDPC techniques presented in Section~\ref{Sec:unifiednew}.} \AC{\rev{In light of the  preceding analysis}}, Section \ref{Sec:new_scheme} introduces $\gamma$-DDPC and discusses the role of regularization in the data-driven control framework. The benchmark numerical example of Section \ref{sec:benchmark} illustrates the effectiveness of the $\gamma$-DDPC perspective in designing a satisfactory control action. The paper is ended by some concluding remarks.

\textbf{Notation.} Matrices will be denoted with capitals (e.g. $A$), column vectors will be denoted with lowercase letters (e.g. $a$). The transpose of $A$ will be denoted with $A^\top$; the notation $A^\dagger$ will denote the Moore-Penrose pseudo-inverse of $A$. Given deterministic (vector) sequences $a(t)$, $b(t)$ the notation $a(t) = O(b(t))$ means that there exist $M$ and $c < \infty$ such that, for all $t>M$, 
\begin{equation}\label{eq:bigO}
\|a(t)\| \leq c \|b(t)\|.
\end{equation}
Similarly we say that $a(t) = o(b(t))$ if $$
\mathop{\rm lim}_{t\rightarrow \infty} \frac{\|a(t)\|}{\|b(t)\|} = 0,$$ or, equivalently, that  for all $\epsilon > 0$, there exists $M < \infty$ such that, for all $t>M$, 
\begin{equation}\label{eq:littleo}
\|a(t)\| \leq \epsilon \|b(t)\|.
\end{equation}
Probabilistic versions of  $O(\cdot)$ and $o(\cdot)$ (i.e., with conditions \eqref{eq:bigO} and \eqref{eq:littleo} holding in probability) will be denoted by $O_P(\cdot)$ and $o_P(\cdot)$, see e.g., \cite{vanderVaart}. \rev{Given $a$ and $b$, we use the symbols $\aseq$ and $\triangleq$ to denote equality up to $o_{P}(1/\sqrt{N})$ and up to $O_P(1/\sqrt{N})$, respectively. Namely
\begin{subequations}
	\begin{align}
		& a \aseq b \iff a=b+o_{P}(1/\sqrt{N}),\\
		& a \triangleq b \iff a=b+O_{P}(1/\sqrt{N}).
	\end{align}
\end{subequations}} 
$\Pi_{A}[B]$ denotes the orthogonal projection of the (rows of the) matrix $B$ on the row span of the matrix $A$, i.e.,
$$
\Pi_{A}[B]  = BA^\top (AA^\top)^\dagger A.
$$
Similarly $\Pi_{A,C}[B]$ indicates the projection of $B$ onto the row span of $A$ and $C$. 
Finally, given a signal $w(k) \in \R^s$, we define the associated Hankel matrix $W_{[t_0,t_1],N} \in \R^{s(t_1-t_0+1) \times N}$ as:
\begin{equation}\label{eq:Hankel}
W_{[t_0,t_1],N}\!:=\!\!\frac{1}{\sqrt{N}}\!\begin{bmatrix}
w(t_0) & w(t_0\!+\!1) & \cdots & w(t_0\!+\!N\!-\!1)\\
w(t_0\!+\!1) & w(t_0\!+\!2) & \cdots & w(t_0\!+\!N)\\
\vdots & \vdots & \ddots & \vdots\\
w(t_1) &w(t_1\!+\!1) & \dots & w(t_1\!+\!N\!-\!1)  
\end{bmatrix}\!\!,
\end{equation}
while the shorthand $W_{t_0}:= W_{[t_0,t_0],N}$ is used to denote the Hankel containing a single row, namely:
\begin{equation}\label{eq:Hankel:onerow}
W_{t_0}:= \frac{1}{\sqrt{N}}\begin{bmatrix}
w(t_0) & w(t_0\!+\!1) & \cdots & w(t_0\!+\!N\!-\!1)
\end{bmatrix}. 
\end{equation}

\section{Setting and goal}\label{Sec:goal}
Consider an \emph{unknown} discrete-time, \emph{linear time-invariant} (LTI) stochastic plant, {whose behaviour can always be described by the so-called \emph{innovation-form}  equations}
\begin{equation}\label{eq:stoc_sys}
	\begin{cases}
		x(t+1)=Ax(t)+Bu(t)+Ke(t)\\
		y(t)=Cx(t)+Du(t)+e(t), 	\end{cases} \quad  t \in \Z
\end{equation}
where $x(t)\in \mathbb{R}^{n}$, $u(t) \in \mathbb{R}^{m}$ and $e(t) \in \mathbb{R}^{p}$ are the state, input and  innovation process  respectively, while $y(t) \in \mathbb{R}^{p}$ is the corresponding  output signal. Without loss of generality we shall assume that \eqref{eq:stoc_sys} is minimal (i.e., reachable and observable).

Given a \rev{constant} reference signal $\rev{y_{r}}$\rev{, a constant reference input $u_r$}
, and a control horizon $T$, the receding horizon predictive control problem can be framed as follows:
\begin{subequations}\label{eq:RHPC_prob}
	\begin{align}
		&\underset{u(k), k \in [t,t+T)}{\mbox{minimize}}~\frac{1}{2}\left[\sum_{k=t}^{t+T-1} \!\!\!\E\left[\|y(k)\!-\! y_{r}
		\|_{Q}^{2}\right]\!+\!\|u(k)\!-\!\rev{u_r}\|_{R}^{2}\right] \label{eq:cost}\\
		& \mbox{s.t. } x(k\!+\!1)\!=\!\!Ax(k)\!+\!Bu(k)\!+\!Ke(k),~k \!\in\! [t,t\!+\!T),\\
		& \qquad  y(k)\!=\!Cx(k)+Du(k)+e(k),~k \in [t,t+T),\\ 		
		&\qquad x(t) = x_{init},\\ 
		&\qquad u(k) \in \mathcal{U},~\E[ y(k)] \in \mathcal{Y},~k \in [t,t+T),
	\end{align}
\end{subequations}
where \rev{{$k \in \mathbb{Z}$}}, $x_{init}$ is the state at time $t$, $e(k)$ is  a zero mean noise with variance $Var\{e(k)\}$, the sets $ \mathcal{U}$, $ \mathcal{Y}$ denote inputs and output constraints, and the expectation $\E[\cdot]$ is taken w.r.t. the future noise sequence $e(k)$, $ k \in [t,t+T)$, and \emph{conditionally} on the initial state $x_{init}$ and the future input trajectory {$u_f: = \{u(k), k \in [t,t+T)\}$.} The tunable symmetric weights $Q \in \mathbb{R}^{p \times p}$ and $R \in \mathbb{R}^{m \times m}$, with $Q \succeq 0$ and $R \succ 0$, have to be selected to trade-off between tracking performance and the required control effort. 
Our goal is to solve problem \eqref{eq:RHPC_prob} when the systems matrices $A,B,C,D,K$ are \emph{not known} and only a sequence of input output data $\mathcal{D}_{N_{data}}=\{u(j),y(j)\}_{j=1}^{N_{data}}$ collected in open loop\footnote{Extension to data collected in closed-loop is possible. Yet, for the sake of exposition, its treatment is deferred to future publications.} from system \eqref{eq:stoc_sys} is available.

\subsection{Features of the predictive control problem} \label{Sec:prelim}
We now elaborate on the optimization problem \eqref{eq:RHPC_prob} and make two important observations:
\begin{enumerate}
\item Problem \eqref{eq:RHPC_prob} can be equivalently formulated only in terms of the so called ``deterministic'' part of the stochastic system \eqref{eq:stoc_sys}, i.e., the one depending only on the control input and the initial state, but not on the noise $e(k)$.
\item The initial state $x_{init}$ at time $t$ does not have to be available. Indeed, it can be accounted for with arbitrary accuracy based on a sufficiently long window of past input-output observations. 
\end{enumerate}

To show that the first point holds, it is useful to rewrite the control problem \eqref{eq:RHPC_prob} exploiting the decomposition of second order moments as the sum of squared means plus variance, i.e.,  
\begin{align*}
	 \E \left[\|y(k) - \rev{y_{r}}\|_{Q}^{2}\right] &= \| \E\left[y(k) \right] - \rev{y_{r}}\|_{Q}^{2}  + \\
	& + \underbrace{\E\left[\| y(k) - \E\left[y(k)\right]  \|_{Q}^{2}\right]}_{\mbox{\scriptsize independent of } u(k)}
\end{align*}
Since the variance term $\E[\|y(k)-\E[y(k)]\|_{Q}^{2}]$ is independent of the input signal $u(k)$, $k \in [t,t+T)$, only the {conditional (given $x_{ini} $ and $u_f$)} mean value of the output, namely $y^d(k):= \E   [y(k)]$ affects the optimization problem. Denoting {with $x^d(k)$ the conditional mean of $x(k)$, i.e. $x^d(k): = \E[x(k)]$,} it is straightforward to see that the optimal control problem \eqref{eq:RHPC_prob} can be equivalently recast as 
\begin{subequations}\label{eq:RHPC_prob_mean}
	\begin{align}
	&\underset{u(k), k \in [t,t+T)}{\mbox{minimize}}~\frac{1}{2}\left[\sum_{k=t}^{t+T-1} \!\!\!\|y^{d}(k)\!-\! \rev{y_{r}}\|_{Q}^{2}\!+\!\|u(k)\!-\!\rev{u_r}\|_{R}^{2}\right] \label{q:cost_det}\\
	& \mbox{s.t. } x^{d}(k\!+\!1)\!=\!\!Ax^{d}(k)\!+\!Bu(k),~k \!\in\! [t,t\!+\!T),\\
	& \qquad  y^{d}(k)\!=\!Cx^{d}(k)+Du(k),~k \in [t,t+T),\\ 		
	&\qquad x^{d}(t) = x_{init},\\ 
	&\qquad u(k) \in \mathcal{U},~y^{d}(k) \in \mathcal{Y},~k \in [t,t+T).
\end{align}	
\end{subequations}
Even though only the ``deterministic'' part of the system influences the optimal control problem, it is important to stress that measured data are indeed affected by noise. This should be accounted for when exploiting measured data $\mathcal{D}_{N_{data}}$ to solve \eqref{eq:RHPC_prob_mean}.

As it concerns the second observation, to prove its validity we exploit the fact that \eqref{eq:stoc_sys} can be written in innovation (or ``whitening'' \cite{ChiusoP-05a}) form. Accordingly, it holds that
\begin{equation}\label{eq:stoc_sys_inn}
	\begin{cases}
		x(k+1)=(A-KC)x(k)+Bu(k)+Ky(k)\\
		e(k)=y(k) - Cx(k)+Du(k),
	\end{cases}
\end{equation}
and, for any $\rho >0$, $\rho \in \Z$, 
\begin{equation}\label{eq:xinit}
	x(t)\!\!=\!\!(A-KC)^{\rho}x(t\!-\!\rho)\!+\!\!\sum_{p=1}^{\rho} \left[\Phi_{p}u(t\!-\!p)\!+\!\Psi_{p}y(t\!-\!p)\right],
 \end{equation}
where $\Phi_{p}=(A-KC)^{p-1}B$ and $\Psi_{p}=(A-KC)^{p-1}K$. By denoting with $\lambda_{max}$ the eigenvalues of $A-KC$ of largest absolute value, under the (mild) assumption that the matrix $A-KC$ is strictly stable, i.e., 
$|\lambda_{max}| < 1$, we have that: 
\begin{equation}\label{eq:state:approx}
 x(t) = \mathcal{C}\begin{bmatrix}
 	u_t^-\\
 	y_t^-
 \end{bmatrix} + \underbrace{ O(|\lambda_{max}|^\rho)}_{\rightarrow 0 \mbox{ for } \rho \rightarrow \infty} 
\end{equation}
where the $O(\cdot)$ term goes   to zero exponentially; $\mathcal{C}$ stacks the (reversed) controllability matrices $\mathcal{C}_{u}$ and $\mathcal{C}_{y}$, i.e., 
\begin{equation*}
	\mathcal{C}=\begin{bmatrix}\mathcal{C}_{u} &
		\mathcal{C}_{y}
		\end{bmatrix}=\begin{bmatrix}
		\Phi_\rho & \cdots & \Phi_2 & \Phi_1 & \Psi_\rho & \cdots & \Psi_2 & \Psi_1
	\end{bmatrix},
\end{equation*}
and 
\begin{equation}\label{eq:past_rho}
	u_t^-\!: =\!\begin{bmatrix} u(t\!-\!\rho) \\ \vdots \\ u(t\!-\!2) \\ u(t\!-\!1) \end{bmatrix}\!\!,\quad y_t^-\!: =\!\begin{bmatrix} y(t\!-\!\rho) \\ \vdots \\ y(t\!-\!2) \\ y(t\!-\!1)\end{bmatrix}
\end{equation}
\rev{are \emph{noisy} collections of past inputs and outputs.} The relation in \eqref{eq:state:approx} thus guarantees that, up to  $O(|\lambda_{max}|^\rho)$ terms, the initial state can be uniquely reconstructed with a finite window of past data.\\

\begin{remark}[State/data relation]\label{rem:deterministic}
	In the so-called ``deterministic case'', i.e., when there is no process/measurement noise in \eqref{eq:stoc_sys}, the state at time $t$ is a (deterministic) function of a \emph{finite} past window of input-output data, As such, $\exists \mathcal{C}_{det}$ such that 
	\begin{equation*}
		x(t)\!=\!\mathcal{C}_{det}\begin{bmatrix} u_t^- \\ y_t^-\end{bmatrix}, 
	\end{equation*}   
	provided $\rho \geq n$. This is a trivial consequence of \eqref{eq:xinit} and of the observability of the system.\\ 
\end{remark}

\begin{remark}[Choice of $\rho$ - part I]\label{rem:rho}
	In (subspace) system identification, see e.g., \cite{Bauer-05,ChiusoAUTO06,ChiusoTAC07,ChiusoTAC2010} {the quantity $\rho$, known also as the ``past horizon'',}  has to be determined from measured data trading off bias and variance. Indeed, $\rho$ should be large, so that the quantity $O(|\lambda_{max}|^\rho)$ can be neglected, but a large $\rho$ ultimately requires estimating larger \rev{sample covariance matrices}. { A simple and effective way of determining $\rho$ in a data-driven fashion is by using Akaike's criterion (e.g., FPE) \cite{Akaike1969}, with the latter choice also guaranteeing that \rev{$\|(A-KC)^\rho\| = O(|\lambda_{max}|^\rho) = o_P(1/\sqrt{N_{data}})$.}  This is in contrast with common practice in the literature of DDPC where the length $\rho$ of the past horizon 
	is not linked to the eigenvalues of $(A-KC)$, but rather it is generally chosen based on (e.g., an upper bound of) the ``order''   $n$ of the deterministic model. \rev{Finally, note that $(A-KC)$ encodes information both on the deterministic dynamics and the noise properties. Hence, the choice of $\rho$ is intimately related to the stochastic nature of the the disturbances, a feature commonly neglected in DDPC schemes. }}
\end{remark}
\section{DDPC formulation via subspace methods}\label{sec:SID}
In this Section, we exploit ideas from subspace identification to recast Problem \eqref{eq:RHPC_prob_mean} in terms of observed input output data $D_{N_{data}}$.  \rev{The results in this section are standard in subspace identification and can be found in several references, see for instance \cite{Bauer-J-00,Dahlen,Bauer-05,ChiusoTAC04,ChiusoAUTO06,ChiusoTAC07}.}

\rev{Let us first define the joint input/output process}
\begin{equation*}
	z(k):=\begin{bmatrix}
		u(k)\\
		y(k)
	\end{bmatrix},
\end{equation*}
and introduce the shorthands for the ``past'' Hankel matrices, namely,
\begin{equation}\label{eq:past}
	U_P\!\!:=\!U_{[0,\rho-1],N},~Y_P\!\!:=\!Y_{[0,\rho-1],N},~ Z_P\!\!:=\!Z_{[0,\rho-1],N}
\end{equation}
and the ``future'' ones, i.e., 
\begin{align}\label{eq:future}
	\nonumber U_F\!:=&U_{[\rho,\rho+T-1],N},~Y_F\!:=\! Y_{[\rho,\rho+T\!-\!1],N},\\
	&\qquad E_F\!:=\!E_{[\rho,\rho+T-1],N} 
\end{align}
Note that, once the lengths of both the ``past'' $\rho$ and ``future'' $T$ are fixed, the number of columns $N$ of the Hankel data matrices is chosen in such a way that all the available data are exploited, namely $N\!:=\!N_{data}\!-\!T\!-\!\rho$. \\
Let us further introduce the extended observability matrix $\Gamma \in \mathbb{R}^{pT\times n}$ associated with the system in \eqref{eq:stoc_sys}, namely
\begin{equation}\label{eq:observability}
	\Gamma=\begin{bmatrix} C \\ CA \\ CA^{2}\\ \vdots \\ CA^{T-1} \end{bmatrix},
\end{equation}   
and the Toeplitz matrices $\mathcal{H}_{d} \in \mathbb{R}^{pT \times mT}$ and $\mathcal{H}_s \in \mathbb{R}^{pT \times pT}$
formed with its Markov parameters, i.e.,
\begin{subequations}
	\begin{align}
		& 	\mathcal{H}_\rev{d} = \begin{bmatrix} 
			D & 0  & 0 & \dots & 0  \\
			CB & D  &  0 &\dots & 0 \\
			CAB & CB & D  &  \dots & 0 \\
			\vdots & \vdots  & \vdots &  \ddots & \vdots & \\
			CA^{T-2}B & CA^{T-3}B  & CA^{T-4}B & \ldots &D 
		\end{bmatrix},\\
		&\mathcal{H}_s = \begin{bmatrix} 
			I & 0  & 0 & \dots & 0 \\
			CK & I  &  0 &\dots & 0 \\
			CAK & CK & I  &  \dots & 0 \\
			\vdots & \vdots  & \vdots &  \ddots & \vdots  \\
			CA^{T-2}K & CA^{T-3}K  & CA^{T-4}K & \ldots &I 
		\end{bmatrix}. 
	\end{align}
\end{subequations}
\rev{Based on  \eqref{eq:state:approx} and {provided $\rho$ is chosen in a data-driven fashion as discussed in Remark~\ref{rem:rho}}, $X_\rho$  can be {written} as 
	\begin{equation}\label{eq:x0}
		X_\rho=\underbrace{{\mathcal C}_u U_P +  {\mathcal C}_y Y_P}_{:= {\mathcal C} Z_P} + \underbrace{(A-KC)^\rho}_{O(|\lambda_{max}|^\rho)} X_0~\aseq~{\mathcal C} Z_P,
	\end{equation}
	where $Z_{P}=\begin{bmatrix}
		U_{P}^\top & Y_{P}^\top
	\end{bmatrix}^\top$.} The Hankel matrix of future outputs $Y_F$ thus satisfies the following:
\begin{equation}\label{eq:Hankel_links}
\rev{\begin{array}{rl}
	Y_{F}&=\Gamma X_{\rho}+\mathcal{H}_{d}U_{F}+\mathcal{H}_{s}E_{F} \\
& \aseq {\mathcal C} Z_P+\mathcal{H}_{d}U_{F}+\mathcal{H}_{s}E_{F} ,
\end{array}}
\end{equation} 
which is the {equation often} considered as a starting point in subspace identification \rev{and control} \cite{OverscheeDeMoor94,ChiusoP-05a,FAVOREEL1999}. We can now characterize the future noise $E_{F}$ according to the following.\\
\begin{lemma}[Projection of noise]\label{Lemma:noise}
	For any fixed $\rho$ in \eqref{eq:past_rho}, it holds that
	\begin{equation}
		\Pi_{Z_{P},U_{F}}(E_{F})= \Upsilon\begin{bmatrix} Z_P \\ U_F\end{bmatrix}\!,
	\end{equation}
	where $\| \Upsilon\| = O_{P}\!\left(\!\frac{1}{\sqrt{N}}\!\right)$ and $Z_{P}$, $U_F$ are   defined as in \eqref{eq:past}and \eqref{eq:future}.
\end{lemma}  
{\begin{proof}
By definition, 
$$\Pi_{Z_{P},U_{F}}(E_{F}) = \underbrace{E_F\left[\begin{matrix} Z^\top_P & U^\top_F \end{matrix} \right] \left[ \left[\begin{matrix} Z_P \\U_F \end{matrix} \right]\left[\begin{matrix}Z^\top_P & U^\top_F \end{matrix} \right]\right]^{-1}}_{\Upsilon} \left[\begin{matrix}Z_P \\U_F \end{matrix} \right],
$$
so that 
$$
\Upsilon = \underbrace{E_F\left[\begin{matrix} Z^\top_P & U^\top_F \end{matrix} \right] }_{[ \hat \Sigma_{e_F z_P} \; \hat \Sigma_{e_F u_F}]}  \left[\left[\begin{matrix} Z_P \\ U_F \end{matrix} \right] \left[\begin{matrix}Z^\top_P & U^\top_F \end{matrix} \right]\right]^{-1}.$$
It is sufficient to observe that the term on the left-hand side $\hat \Sigma_{e_F z_P}$ and $\hat \Sigma_{e_F u_P}$ are sample cross-covariances between future innovations and past data ($z_P$) or future inputs $u_F$, and thus converge to zero in probability with rate $\frac{1}{\sqrt{N}}$, whereas the rightmost term converges to the input-output covariance matrix,  which is bounded away from zero \rev{thanks to the persistency of excitation assumptions, see, e.g., Lemma \ref{lem:PE} later.}
\end{proof}
}
This Lemma further allows us to characterize the future outputs $Y_{F}$ as follows.\\
\begin{lemma}[Projection of the output]\label{Lemma:output}
	The projection $\hat Y_F\!:=\!\Pi_{Z_{P},U_{F}}(Y_F)$ satisfies 
{	\begin{align}
		\nonumber \hat{Y}_{F}&=\Gamma \hat X_{\rho} +\mathcal{H}_{d}U_{F}+ \mathcal{H}_{s}\Pi_{Z_{P},U_{F}}(E_F)\\
		&\triangleq \Gamma {\mathcal C} Z_P +  \mathcal{H}_{d}U_{F}, \label{eq:Proj:F}
	\end{align}
	}
	where $ \hat X_{\rho} : = \Pi_{Z_{P},U_{F}}(X_{\rho}) ~\aseq~{\mathcal C} Z_P$.
\end{lemma}
{\begin{proof}
The proof {straightforwardly} follows from the observation that the projection is a linear operator, by exploting Lemma \ref{Lemma:noise} on the projection of the noise term and Equation \eqref{eq:x0} on the approximation of the state using a finite set of past data. 
\end{proof}}
Given the {projected} initial condition $\hat X_{\rho}$ and the input $U_F$, Lemma~\ref{Lemma:output} establishes that the projected output $\hat Y_F$ equals the evolution of the deterministic part of the system \eqref{eq:stoc_sys}, up to $O_P(1/\sqrt{N})$ terms. This result is formalized in the following Theorem.\\

{\begin{theorem}[Output/data relation]\label{thm:DDoutputs}
Given any $\alpha \in \R^N$, the vector $\hat y^d_f: = \hat Y_F \alpha$ satisfies the relation 
\begin{equation}\label{eq:fut_out}
\begin{array}{rcl}
\hat y^d_f  &= &  \Gamma \hat x^d(t) +  \mathcal{H}_{d} u_f +O_{P}\!\left(\!\frac{1}{\sqrt{N}}\!\right) \\
&\triangleq &  \Gamma \hat x^d(t) +  \mathcal{H}_{d} u_f,
\end{array}
\end{equation}
where 
{\begin{subequations}
\begin{align}
		&\hat x^d(t) :=  \hat X_{\rho} \alpha ~\aseq~ {\mathcal C} Z_P \alpha=  {\mathcal C} z_{init},\label{eq:state}\\
		&u_f := \begin{bmatrix}
			u(t)\\
			u(t+1)\\
			\vdots\\
			u(t+T-1)
		\end{bmatrix}=U_F \alpha, \label{eq:uf}
\end{align}
\end{subequations}}
and $z_{init} := Z_P \alpha$.
\end{theorem}
\begin{proof}
The proof is an immediate consequence of Lemma \ref{Lemma:output}. In fact, defining $\hat y^d_f: = \hat Y_F \alpha$ and using Equation \eqref{eq:Proj:F},
we have that 
$$
\hat y^d_f: = \hat Y_F \alpha \triangleq \Gamma {\mathcal C} \underbrace{Z_P  \alpha}_{:= x_{init}} +  \mathcal{H}_{d}\underbrace{U_{F}  \alpha}_{:= u_{f}}.
$$
\end{proof}
}
The result in Theorem~\ref{thm:DDoutputs} should be read as follows. If the sequence of past input-output $u(k)$ and $y(k)$  for $k \in [t-\rho,t-1]$ equals $z_{init}$  and the future inputs $u(k)$ in the time window $k \in [t,t+T-1]$ (see \eqref{eq:uf}) are given by  $u_{f}$,  the corresponding ``deterministic" output, i.e.,
\begin{equation*}
y_f^d:=\begin{bmatrix}
	y^d(t) \\ y^d(t + 1) \\ \vdots \\ y^d(t+T-1)
\end{bmatrix},
\end{equation*}
is a linear transformation through $\alpha$ of the projected future outputs $\hat Y_F$, up to $O_P(1/\sqrt{N})$ terms. 

\subsection*{Towards DDPC}
For every pair of initial conditions and future inputs that can be written as linear combinations of $Z_P$ and $U_F$ (see \eqref{eq:past} and \eqref{eq:future}), Theorem~\ref{thm:DDoutputs} shows that one can compute the deterministic output of \eqref{eq:stoc_sys} (up to $O_P(1/\sqrt{N})$ terms) from a \emph{finite set}  input-output data only, \emph{without knowing} the true system \eqref{eq:stoc_sys}. {Under the additional assumption that the training input $u(t)$ has a full rank spectral density matrix and the innovation process has positive definite variance $Var\{e(t)\} >0$ \footnote{Since the our purpose \emph{is not} to discuss the weakest conditions under which the results of Theorem~\ref{thm:DDoutputs} can be generalized, here we make a sufficient assumption that is general enough for being widely applicable in practice.}, we can guarantee that the matrices $Z_P$ and $U_F$ have full rank, so that \emph{any}     possible initial condition and sequence of control inputs can be generated by linear combination of their columns. The following lemma formalizes this result.\\
\begin{lemma}[Persistency of excitation]\label{lem:PE}
	If the input process has full rank spectral density that is bounded away from zero and  $Var\{e(t)\} >0$, then for \emph{any} choice of $\rho$ and $T$ and provided $N > (m+p)(\rho+T)$, the block Hankel matrix 
	\begin{equation}\label{eq:Hankeldatamatrix}
		Z_{data}: = \begin{bmatrix} Z_P \\ U_F \\ Y_F \end{bmatrix} \in \R^{(m+p)(\rho+T) \times N}
	\end{equation}  has full rank almost surely. 
\end{lemma}
\begin{proof}
The proof is a direct consequence of the fact that, under the stated assumptions, the joint spectral density matrix of the input-output process $z(t):=[u^\top(t) \; y^\top(t)]^\top$ does not vanish on the unit circle and, therefore, the intersection between the (joint) past and input spaces contains only the zero random variable (see e.g. \cite{Hannan}).  Thus, the Hankel matrix formed with input output trajectories has full rank almost surely.
\end{proof}}
Under the latter, the result in Theorem~\ref{thm:DDoutputs} can be generalized to all initial conditions and future inputs, as stated in the main result of this Section.\\

\begin{theorem}[Output/data relation - generalized]\label{thm:WL:Stoch}
	Under the assumptions in Lemma \ref{lem:PE},  given any (past) joint input and output trajectory 
	\begin{equation}\label{eq_zini}
		z_{init}: = 
		\begin{bmatrix} z(t-\rho) \\ \vdots \\ z(t-2) \\ z(t-1) \end{bmatrix},
	\end{equation}
	and any choice of the future control input 
	\begin{equation}\label{eq_uf}
		u_f: = \begin{bmatrix} u(t) \\ u(t + 1) \\ \vdots \\ u(t+T-1)\end{bmatrix},
	\end{equation} 
	the corresponding ``deterministic'' output 
	$$
	y_f^d:= \begin{bmatrix} y^d(t) \\ y^d(t + 1) \\ \vdots \\ y^d(t+T-1)\end{bmatrix}
	$$ satisfies: 
	\begin{equation}\label{eq:approx:output}
		y_f^d = \hat Y_F \alpha^\star + O_P(1/\sqrt{N})  \triangleq \hat Y_F \alpha^\star
	\end{equation}
	where $\alpha^\star$ is the minimum-norm solution of the  system of linear equations:
	\begin{equation}\label{eq:opt_alpha}
		\begin{bmatrix} z_{ini\rev{t}} \\ u_f \end{bmatrix} = 
		\begin{bmatrix} Z_P \\ U_F\end{bmatrix} \alpha
	\end{equation}
	{where $\hat Y_F:=  \Pi_{Z_{P},U_{F}}(Y_F)$.}
\end{theorem}
{\begin{proof}
Under the assumption of Lemma \ref{lem:PE}, the matrix $Z_{data}$ has full rank and, therefore, $\forall$ $z_{init}$ and $u_f$, there exists $\alpha$ such that 
\begin{equation}\label{eq:init_uf}
\begin{bmatrix} z_{init} \\ u_f  \end{bmatrix} =  \begin{bmatrix} Z_P \\ U_F  \end{bmatrix} \alpha.
\end{equation}
Thus, exploiting Theorem \ref{thm:DDoutputs}, the corresponding deterministic output satisfies 
$$
y_f^d = \hat Y_F \alpha + O_P(1/\sqrt{N})  \triangleq \hat Y_F \alpha.
$$
This is true for \emph{all} possible solutions of \eqref{eq:init_uf}, and in particular it holds for its minimum-norm solution $\alpha^*$.\\ 
\end{proof}}

\begin{remark}[The case of deterministic systems]\label{rem:initialstate:solution}
		The reader may observe that, when $e(t)=0$ (that is the system is actually deterministic), Lemma \ref{lem:PE} does not hold. Indeed, for for $\rho > n$, it is well know (see, e.g., \cite{Moonen}) that the Hankel matrix $Z_{data}$ in \eqref{eq:Hankeldatamatrix}  and $Z_{P} \in  \R^{(m+p)\rho \times N}$ in \eqref{eq:opt_alpha} have rank equal to 
		\begin{align*}
			& \mbox{rank}(Z_{data})=n +   m(\rho+T) < (m+p)(\rho+T),\\
			& \mbox{rank}(Z_P)=n +  m\rho <(m+p)\rho.
		\end{align*}
		{These relations are indeed the basis for the so-called ``intersection algorithms'' in subspace identification, and also can be seen as algebraic formulations of the well known ``Willems' fundamental lemma'' \cite{willems2005note}.}
		Nonetheless, in this case, any finite (deterministic) trajectory $z_{init}$ of the system \eqref{eq:stoc_sys} belongs to the column span of $Z_P$. As such, provided that $z_{init}$ is an ``admissible'' sequence of input/output pairs of the given deterministic system, then \eqref{eq:opt_alpha} has a solution. 
	\end{remark}

\rev{Based on the previous results, w}e are now ready to recast the control problem \eqref{eq:RHPC_prob_mean} in a data driven fashion as {follows}:
\begin{subequations}\label{eq:DD_RHPC_prob_mean}
	\begin{align}
		&\underset{\rev{u_{f}}
		}{\mbox{minimize}}~~	J\left(\begin{bmatrix}
			y_{f}^{d}\\
			u_{f}
		\end{bmatrix}\right)\\
		& \mbox{s.t. }  \alpha^\star=\begin{bmatrix}
			Z_{P}\\
			U_{F}
		\end{bmatrix}^{\dagger}\begin{bmatrix}
			z_{init}\\
			u_{f}
		\end{bmatrix},\label{eq:prediction_model_DD1}\\
		& \qquad y_{f}^{d}=\hat{Y}_{F}\alpha^{\star},\label{eq:prediction_model_DD2}\\
		& \qquad u(k) \in \mathcal{U},~y^{d}(k) \in \mathcal{Y},~k \in [t,t+T), \label{eq:constraints_DD}	 
	\end{align}
	where
	\begin{equation}\label{eq:cost_DD}
		J\left(\!\begin{bmatrix}
			y_{f}^{d}\\
			u_{f}
		\end{bmatrix}\!\right)\!\!=\! \frac{1}{2}\left[\sum_{k=t}^{t+T-1}\!\| y^d(k)- \rev{y_{r}}
		\|_{Q}^{2}+\|u(k)\!-\!\rev{u_r}\|_{R}^{2}\right]\!,
	\end{equation}
\end{subequations}
and $\hat Y_F: = \Pi_{Z_P,U_F}(Y_F)$, while $z_{init}$ and $u_f$ are defined as in \eqref{eq_zini} and \eqref{eq_uf}. \rev{Except for the use of a slightly different notation and the introduction of constraints, the problem in \eqref{eq:DD_RHPC_prob_mean} corresponds to the \emph{Stochastic Predictive Control} (SPC) problem firstly formalized in \cite{FAVOREEL1999}}.

\rev{\section{A unified outlook on DDPC problems}\label{Sec:unifiednew}}
 Recent papers have discussed problems \rev{that are very similar to \eqref{eq:DD_RHPC_prob_mean}, generally starting from a \emph{deterministic} viewpoint, i.e., assuming that $e(t) = 0$, $\forall t$ in \eqref{eq:stoc_sys}, and then coping with measurement noise by introducing slack variables and regularization terms. 
 	However, by reformulating these problems with our notation, we will show that \textit{all of them can be cast into a unified framework}. In particular, we will show the connections among the problem in \eqref{eq:DD_RHPC_prob_mean}, the one with slacks on the initial conditions proposed in \cite[Section IV.B]{Felix2021}, the formulation tailored to cope with bounded noise introduced in \cite{berberich2020data} and that with elastic net regularization given in \cite[Section IV.D]{dorfler2022bridging}.}  

\rev{\subsection{SPC with slacks}}
\rev{Based on our notation, the SPC problem tackled in \cite{Felix2021} to handle non-deterministic scenarios (with measurement noise only) can be recast as follows:}
\rev{\begin{subequations}\label{eq:DD_Lucia}
	\begin{align}
	&\underset{u_{f},\sigma \succeq 0}{\mbox{minimize}}~~	J\left(\begin{bmatrix}
			y_{f}^{d}\\
			u_{f}
		\end{bmatrix}\right)+\lambda\|\sigma\|_{2}^{2} \label{eq:DD_Lucia_cost}\\
		& \mbox{s.t. }  \alpha=\begin{bmatrix}
			Z_{P}\\
			U_{F}
		\end{bmatrix}^{\dagger}\left(\begin{bmatrix}
			z_{init}\\
			u_{f}
		\end{bmatrix}+\begin{bmatrix}
		\sigma\\
		0
	\end{bmatrix}\right),\label{eq:prediction_model_Lucia}\\
		& \qquad y_{f}^{d}=\hat{Y}_{F}\alpha,\\
		& \qquad u(k) \in \mathcal{U},~y^{d}(k) \in \mathcal{Y},~k \in [t,t+T), 
	\end{align}
\end{subequations}
where $\sigma \in \mathbb{R}^{\rho(m+p)}$ is a slack variable to be optimized in order to cope with noise on the data used to build $z_{init}$, while $\lambda>0$ is a tunable parameter\footnote{\rev{By considering a diagonal matrix $\Lambda$ rather than a scalar $\lambda$, different weights can be chosen for the slack acting on past inputs and outputs, like in the framework proposed in \cite{Felix2021}.}}.}
\rev{The role of the additional slack $\sigma$ introduced in \eqref{eq:DD_Lucia} is linked to the results presented in Section~\ref{Sec:prelim} by the following proposition.
\rev{\begin{lemma}[Asymptotic regularization with slacks]
Assume that the cost $J(\cdot)$ in \eqref{eq:DD_Lucia_cost} is equal to \eqref{eq:cost_DD}. Then the solution to problem \eqref{eq:DD_RHPC_prob_mean} coincides with the one of \eqref{eq:DD_Lucia} when $\lambda \to \infty$. Moreover, $\lambda \to \infty$ is the optimal choice even for finite, but large, $N_{data}$ and $\rho$ chosen according to the Akaike's criterion.     
\end{lemma}
\begin{proof}
	The proof of the first statement is a direct consequence of the formulations of the problems. The second claim straightforwardly follows from \eqref{eq:state:approx} and Remark~\ref{rem:rho}, showing that the error due to finite past is $o_P(1/\sqrt{N_{data}})$ (and thus can be neglected), whereas the error due to finite $N_{data}$ in the projection \eqref{eq:approx:output} is instead $O_P(1/\sqrt{N_{data}})$.
\end{proof}
}}

\rev{This additional insight on problem \eqref{eq:DD_Lucia} provides a direct connection between $\rho$ and $N_{data}$ and the value of the slack variables needed to counteract the effect of noise on the initial conditions. Indeed, when $\rho$ is selected according to Akaike's criterion, Remark~\ref{rem:rho} directly links the accuracy of the reconstructed state with the dimension of the dataset. As such,  for large $N_{data}$ , then $\lambda \to +\infty$ is the optimal choice. This will be confirmed by the simulation results reported in Section  \ref{sec:benchmark} (see Figure \ref{Fig:luciaJ}).}
\rev{\subsection{DDPC with bounded measurement noise}}
\rev{Let us now focus on a stochastic settings in which $K=0$ in \eqref{eq:stoc_sys_inn}, the measurement noise $e(t)$ is bounded, namely $\|e(t)\|_{\infty}\leq \bar{\varepsilon}$, and $\bar{\varepsilon}$ is assumed to be \emph{known}. In our framework, the regularized problem proposed in \cite{berberich2020data} to tackle this scenario can be rewritten as follows: 
\begin{subequations}\label{eq:DD_Berberich}
	\begin{align}
		& \underset{u_f, y_f^d, \alpha,\sigma}{\mbox{minimize}}~~~ J\left( \begin{bmatrix} y_f^d \\ u_f \end{bmatrix}\right)  + \lambda_\alpha\bar{\varepsilon}\|\alpha\|_2^{2}+ \lambda_\sigma\| \sigma \|_2^{2}
		\label{eq:cost_Berberich}\\
		& \mbox{s.t. } \begin{bmatrix}
			z_{init}+\mathbbm{1}_{y}\sigma_{init} \\u_{f}\\y_{f}^{d}+\sigma_{y}			
		\end{bmatrix}=\begin{bmatrix}
			Z_{P}\\U_{F}\\Y_{F}
		\end{bmatrix}\alpha,\label{eq:modelBer}\\
		& \qquad \begin{bmatrix}
			u(k)\\
			y^{d}(k)
			\end{bmatrix}=\begin{bmatrix}
			u_{r}\\
			y_{r}
		\end{bmatrix}\!,~~k \in [t\!+\!T\!-\!\rho,t\!+\!T),\label{eq:terminalBer}\\
		& \qquad u(k) \in \mathcal{U},~y^{d}(k) \in \mathcal{Y},~k \in [t,t+T),\\
		& \qquad \|\sigma\|_{\infty} \leq \bar{\varepsilon}(1+\|\alpha\|_{1}), ~k \in [t,t+T), \label{eq:noiseBer}
	\end{align}
where 
\begin{equation}
	\sigma=\begin{bmatrix}
		\sigma_{init}\\
		\sigma_{y}
	\end{bmatrix}  \in \mathbb{R}^{p(T+\rho)},
\end{equation}
is a vector of slacks accounting for the noise acting on the output measurements used to build $Z_{P}$ and $Y_{F}$, while $\mathbbm{1}_{y}$ is a selector function, introduced to add the slack only on the initial outputs comprised in $z_{init}$. Note that, the constraint in \eqref{eq:terminalBer} is a terminal ingredient introduced to guarantee practical stability and recursive feasibility of the DDPC scheme, and the inequality in \eqref{eq:noiseBer} is a \emph{non-convex}\footnote{\rev{The constraint in \eqref{eq:noiseBer} cannot be enforced without resorting to a non-convex optimization routine. Thus, the entity of the slack is practically contained by a proper tuning of $\lambda_{\sigma}$.}} constraint that connects the slack variables to the known features of the measurement noise.
\end{subequations}}

\rev{
	By recasting the robust DDPC formulation in \cite{berberich2020data} within our framework, we now establish its relationship with the constrained SPC scheme in \eqref{eq:DD_RHPC_prob_mean} through the following result, {which shows that the formulation in \cite{berberich2020data}  complemented with an additional constraint on the slack variable $\sigma_y$ is equivalent to constrained SPC, for a suitable chioice of regularization parameters.}
\begin{theorem}[Regularization with bounded noise]\label{thm:DDberb}
	Let the cost $J(\cdot)$ in \eqref{eq:DD_Berberich} be defined as in \eqref{eq:cost_DD} and the non-convex constraint in \eqref{eq:noiseBer} be neglected. Assume that the SPC problem in \eqref{eq:RHPC_prob_mean} is augmented with the terminal ingredient in \eqref{eq:terminalBer}. {Then,  under the additional constraint on the slack variable $\sigma_y = Y_F (I - \Pi) \alpha$, where $\Pi$ is the orthogonal projector onto the column span of $\begin{bmatrix} Z^\top_P  & U^\top_F \end{bmatrix}$, i.e.,
\begin{equation}\label{eq:pi}
	\Pi := \begin{bmatrix} Z^\top_P  & U^\top_F \end{bmatrix}\begin{bmatrix} Z_P  \\ U_F \end{bmatrix}^{\dagger},
\end{equation}
}the solutions of \eqref{eq:DD_RHPC_prob_mean} with terminal constraints and \eqref{eq:DD_Berberich} coincide for $\lambda_{\alpha}\bar{\varepsilon}=0$ and $\lambda_{\sigma} \to \infty$. 
\end{theorem}
\begin{proof}
		For $\lambda_{\sigma} \to \infty$ and $\lambda_{\alpha}\bar{\varepsilon}=0$, \eqref{eq:DD_Berberich} reduces to
	\begin{subequations}\label{eq:DD_Berberich2}
		\begin{align}
			& \underset{u_f, y_f^d, \alpha,\sigma}{\mbox{minimize}}~~~ J\left( \begin{bmatrix} y_f^d \\ u_f \end{bmatrix}\right) 
			\label{eq:cost_Berberich2}\\
			& \mbox{s.t. }~ \begin{bmatrix}
				z_{init}+\mathbbm{1}_{y}\sigma_{init} \\u_{f}\\y_{f}^{d}+\sigma_{y}			
			\end{bmatrix}=\begin{bmatrix}
				Z_{P}\\U_{F}\\Y_{F}
			\end{bmatrix}\alpha,\\
			& \qquad \begin{bmatrix}
				u(k)\\
				y^{d}(k)
			\end{bmatrix}=\begin{bmatrix}
				u_{r}\\
				y_{r}
			\end{bmatrix}\!,~~k \in [t\!+\!T\!-\!\rho,t\!+\!T),\label{eq:terminalBer2}\\
			& \qquad u(k) \in \mathcal{U},~y^{d}(k) \in \mathcal{Y},~k \in [t,t+T),\\
			& \qquad \sigma =0, \label{eq:slackBer}
		\end{align}
	\end{subequations}	
	where the constraint in \eqref{eq:noiseBer} can be replaced with \eqref{eq:slackBer}, independently from $\bar{\varepsilon}$. Let us now decompose the prediction model in \eqref{eq:modelBer} as follows:
	\begin{align}
		& \begin{bmatrix}
			z_{init}\\
			u_{f}
		\end{bmatrix}+\begin{bmatrix}
		\mathbbm{1}_{y}\sigma_{init}\\
		0
	\end{bmatrix}=\begin{bmatrix}
	Z_{P}\\U_{F}
\end{bmatrix}\alpha, \label{eq:pastANDinputs}\\
		& y_{f}^{d}+\sigma_{y}=Y_{F}\alpha, \label{eq:futureANDslack}
	\end{align} 
	and $Y_{F}$ as $Y_{F}:=\hat{Y}_{F}+\bar{Y}_{F}$, with $\hat{Y}_{F}=\Pi_{Z_{P},U_{F}}(Y_{F})$ and $\bar{Y}_{F}=Y_{F}-\Pi_{Z_{P},U_{F}}(Y_{F})$. Leveraging on \ref{eq:slackBer},  {and the additional constraint $$
\sigma_y =  Y_F (I - \Pi) \alpha = 0,$$ the relations in \eqref{eq:pastANDinputs} and \eqref{eq:futureANDslack} can be rewritten as:
	\begin{align}
		\alpha &=\begin{bmatrix}
			Z_{P}\\
			U_{F}
		\end{bmatrix}^{\dagger}\begin{bmatrix}
		z_{init}\\
		u_{f}
	\end{bmatrix}, \label{eq:SOLpastANDinputs}\\
		 y_{f}^{d}&=Y_{F}\alpha =  Y_{F}\Pi\alpha + Y_{F}(I-\Pi)\alpha =  Y_{F}\Pi\alpha. \label{eq:SOLfutureANDslack}
	\end{align} 
Note that \eqref{eq:SOLpastANDinputs} corresponds to \eqref{eq:prediction_model_DD1} in the constrained SPC problem, whereas \eqref{eq:SOLfutureANDslack} implies that the predictor $y_{f}^{d}=\hat{Y}_{F}\alpha$, thus concluding the proof. }
\end{proof}
}

\rev{The relationship established in Theorem~\ref{thm:DDberb} are consistent with the empirical evidences on the role of $\lambda_{\alpha}\bar{\varepsilon}$ and $\lambda_{\sigma}$ discussed in \cite[Section V]{berberich2020data}. Although providing a guideline for the choice of these two hyperparameters, with the choice of $\lambda_{\alpha}$ inherently connected with the noise bound $\bar{\varepsilon}$, it is worth stressing that this choice will be optimal asymptotically, i.e., when $\rho$ is selected according to the Akaike's criterion and $N_{data} \rightarrow \infty$.}
\rev{\subsection{DeePC with elastic net regularization}}
We \rev{now} consider the  problem with elastic net regularization in \cite[Section IV.D]{dorfler2022bridging}, that we rewrite for the control problem considered in this work by using our notation as follows:
\begin{subequations}\label{eq:DD_Dorfler}
	\begin{align}
		& {\underset{u_f, y_f^d, \alpha}{\mbox{minimize}}~~~ J\left( \begin{bmatrix} y_f^d \\ u_f \end{bmatrix}\right)\!+\! \lambda_1\|\alpha\|_1\!+\!\lambda_2\| (I\!-\!\Pi)\alpha \|_p}
		\label{eq:cost_Dorfler}\\
		& \mbox{s.t. } \begin{bmatrix}
			z_{init}\\u_{f}\\y_{f}^{d}			
		\end{bmatrix}=\begin{bmatrix}
		Z_{P}\\U_{F}\\Y_{F}
	\end{bmatrix}\alpha,\\
	& \qquad u(k) \in \mathcal{U},~y^{d}(k) \in \mathcal{Y},~k \in [t,t+T),
	\end{align}
where $\Pi$ \rev{has been defined in \eqref{eq:pi}.}
\end{subequations}
The following proposition provides the connection between the problem in \eqref{eq:DD_Dorfler} and the one in \eqref{eq:DD_RHPC_prob_mean}.
\begin{theorem}[\rev{SPC-based} regularization]\label{thm:comparisonDorfler}
		Assuming the cost $J(\cdot)$ in \eqref{eq:cost_Dorfler} is equal to \eqref{eq:cost_DD}, then the solution to problem \eqref{eq:DD_RHPC_prob_mean} coincides with the one of \eqref{eq:DD_Dorfler} for $\lambda_1 = 0$ and $\lambda_2 \to \infty$.
\end{theorem}
\begin{proof}
{For $\lambda_1=0$ and $\lambda_2 \to +\infty$, problem \eqref{eq:DD_Dorfler} reduces to 
\begin{subequations}\label{eq:DD_Dorfler/_reduced}
	\begin{align}
		& {\underset{u_f, y_f^d, \alpha}{\mbox{minimize}}~~~ J\left( \begin{bmatrix} y_f^d \\ u_f \end{bmatrix}\right)  }
		\label{eq:cost_Dorfler2}\\
		& \mbox{s.t. } \begin{bmatrix}
			z_{init}\\u_{f}\\y_{f}^{d}			
		\end{bmatrix}=\begin{bmatrix}
		Z_{P}\\U_{F}\\Y_{F}
	\end{bmatrix}\alpha, \quad  \| (I-\Pi)\alpha \| = 0\\
	& \qquad u(k) \in \mathcal{U},~y^{d}(k) \in \mathcal{Y},~k \in [t,t+T).
	\end{align}
\end{subequations}
In addition, by decomposing $Y_F: = \hat Y_F + \tilde Y_F$,
where $\hat Y_F = \Pi_{Z_P,U_F}(Y_F)$, and $\tilde Y_F = Y_F - \Pi_{Z_P,U_F}(Y_F)$, we have that
$$
\hat Y_F = Y_F \Pi  \quad  \tilde Y_F =Y_F (I-\Pi).
$$
Then, when  $(I-\Pi) \alpha = 0$, we have $$Y_F \alpha = Y_F \Pi \alpha + Y_F(I- \Pi) \alpha = Y_F \Pi \alpha  = \hat Y_F \alpha.$$}
\end{proof}
This result not only shows the connection between the control problem considered in this work and the regularized one proposed in \cite{dorfler2022bridging}, but it also puts the results shown in  \cite{dorfler2022bridging}, where the role of $\lambda_1$ and $\lambda_2$ is evaluated experimentally, into a rigorous frame. \rev{Note that, since the performance of SPC is influenced by the choice of $\rho$ in \eqref{eq:state:approx} and the dimension of the Hankel matrix (specifically the number of its columns $N$, see \eqref{eq:fut_out}), the choice of the regularization weights via Theorem~\ref{thm:comparisonDorfler} is likely to be optimal when $\rho$ is chosen according to the Akaike's criterion and $N_{data} \rightarrow \infty$ (i.e., $N \rightarrow \infty$).}

	\section{The $\gamma$-DDPC scheme}\label{Sec:new_scheme}
	
	In this Section, we reformulate problem \eqref{eq:DD_RHPC_prob_mean} by exploiting the LQ decomposition of the Hankel data matrices. On the one hand, this procedure leads to an even closer connection with subspace identification. On the other, it allows us to parametrize the solution to   \eqref{eq:DD_RHPC_prob_mean} in terms of a lower dimensional parameter vector. 
	We thus consider the LQ decomposition of the joint input-output block Hankel matrix  $Z_{data}$ in \eqref{eq:Hankeldatamatrix}, namely:
	\begin{equation}\label{eq:LQ}
	\begin{bmatrix}
	Z_{P}\\
	U_{F}\\
	Y_{F}
	\end{bmatrix}=\begin{bmatrix}
	L_{11} & 0 & 0  \\
	L_{21} & L_{22} &  0\\
	L_{31} & L_{32} & L_{33} 
	\end{bmatrix}\begin{bmatrix}
	Q_{1}\\
	Q_{2}\\
	Q_{3}
	\end{bmatrix},
	\end{equation}
	where the matrices $\{L_{ii}\}_{i=1}^{3}$ are all non-singular (under the assumptions of Lemma \ref{lem:PE}) and $Q_{i}$ have orthonormal rows, i.e.  \emph{i.e.,} $Q_{i}Q_{i}^{\top}=I$, for $i=1,\ldots,3$, $Q_i Q_j^\top = 0$, $i\neq j$. 
	
	First of all, let us observe that $\hat Y_F: = \Pi_{Z_{P},U_{F}}(Y_F)$ in Lemma \ref{Lemma:output} can be expressed in terms of the LQ decomposition \eqref{eq:LQ} as:
	\begin{equation}\label{eq:Yhat:LQ}
	\hat Y_F = \begin{bmatrix}
	L_{31} & L_{32} 
	\end{bmatrix}\begin{bmatrix}
	Q_{1}\\Q_{2}
	\end{bmatrix}.
	\end{equation}
	By exploiting \eqref{eq:LQ} and \eqref{eq:Yhat:LQ}, we can thus express the constraint in \eqref{eq:opt_alpha} as follows:
	\begin{subequations}\label{eq:decomposition}
		\begin{align}
		& z_{init}
		=Z_{P}\alpha= L_{11} Q_{1} \alpha
		\label{eq:init_cond} \\
		&	u_{f} = U_F \alpha  
		=\begin{bmatrix}
		L_{21} & L_{22} 
		\end{bmatrix}\begin{bmatrix}
		Q_{1}\\Q_{2}	
		\end{bmatrix}\alpha, \label{eq:input}
		\end{align}
	\end{subequations}
	where \eqref{eq:init_cond} accounts for the initial condition  of the predictive control problem, whereas   \eqref{eq:input} links the optimal $\alpha$ with the control input. The predicted output in \eqref{eq:approx:output} can then be rewritten as
	\begin{equation}
	\hat y^d_{f} = \hat Y_F \alpha  
	=\begin{bmatrix}
	L_{31} & L_{32} 
	\end{bmatrix}\begin{bmatrix}
	Q_{1}\\Q_{2}	
	\end{bmatrix}\alpha^\star, \label{eq:pred:output}
	\end{equation}
	where $\alpha^\star$ is the minimum-norm solution to \eqref{eq:decomposition}.
	
	We can now leverage on triangular structure of \eqref{eq:decomposition} to characterize the minimum-norm solution $\alpha^\star$.  In particular, \eqref{eq:init_cond} always admits a solution (see Lemma \ref{lem:PE} and Remark \ref{rem:initialstate:solution}), that satisfies the following property. 
	\begin{lemma}[Definition of $\gamma_1$]\label{Lemma:4}
		Let $\alpha^\star_{init} \in \mathbb{R}^{N}$  be the minimum-norm $\alpha$ solving \eqref{eq:init_cond}. Then, by defining $\gamma_1^\star  \in \mathbb{R}^{(m+p)\rho}$ as the unique solution of 	
		\begin{equation}\label{eq:init_problem}
		z_{init} = L_{11} \gamma_1,
		\end{equation}
		$\alpha^\star_{init}$ can be written as 
		$\alpha^\star_{init}=Q_{1}^\top \gamma_{1}^\star $, 
		so that				
		\begin{equation*}
		\alpha^\star_{init} \in \mbox{colspan}\left(Q_1^\top\right).
		\end{equation*}	 
	\end{lemma}
	\begin{proof}
		{ Since $Z_P$ has full column rank, so does $L_{11}$ and any solution $\alpha$ to \eqref{eq:init_cond} must satisfy
		$$
		Q_{1}\alpha = L_{11}^{-1} z_{init} = \gamma_1^*.$$ The minimum-norm solution $\alpha_{init}^*$ can be found by as 
		$$
		\alpha_{init}^* =  Q_{1}^{\dagger}  \gamma_1^* = Q_1^\top \gamma_1^*,
		$$ thus concluding the proof.}
	\end{proof}
	Exploiting the definition of $\gamma_1^\star$ in Lemma \ref{Lemma:4}, the control sequence $u_{f}$ in \eqref{eq:input} can be equivalently written as
	\begin{equation}\label{eq:u_fbis}
	\begin{array}{rcl}
	u_{f}&=&L_{21}\gamma^\star_{1}+L_{22}\gamma_{2}\\
	\gamma_2 & = & Q_{2}\alpha
	\end{array}
	\end{equation}
	Based on this representation, we can provide additional insights on $\alpha^{\star}$, through the following result.  
	\begin{lemma}[Definition of $\gamma_2$]\label{Lemma:5}
		Let $\alpha^\star_{f} \in \mathbb{R}^{N}$ indicate the minimum-norm $\alpha$ solving \eqref{eq:u_fbis}. Accordingly, define $\gamma^*_{2} \in \mathbb{R}^{mT}$ as the unique solution of the least squares problem
		\begin{equation}\label{eq:stepalgo}
		L_{22}\gamma_{2}=u_{f}-L_{21}\gamma^\star_{1}		\end{equation}
		where $\gamma^\star_{1}$ is defined in Lemma \ref{Lemma:4}. Then $\alpha^\star_{f}$ can be written as: 
		\begin{equation}\label{eq:gamma3}
		\alpha^\star_{f} =Q_{2}^\top \gamma^\star_{2},
		\end{equation}
		so that 		
		\begin{equation*}
		\alpha^\star_{f} \in \mbox{colspan}\left(Q_2^\top\right).
		\end{equation*}	
	\end{lemma}
	\begin{proof}
		{Since the matrix $[Z_P^\top \; U_F^\top]^\top$ has full rank, also $L_{22}$ has full rank and is thus invertible.
		Any solution $\alpha$ to 
		$$
		u_f = L_{21} Q_1 \gamma_1^* + L_{22} Q_2 \alpha
		$$ must therefore satisfy
		$$
		Q_2 \alpha = \underbrace{L_{22}^{-1} \left[u_f - L_{21} \gamma_1^*\right]}_{:= \gamma_2^*}.
		$$
		and the minimum-norm solution is given by 
		$$\alpha_f^* = Q_2 ^{\dagger} L_{22}^{-1}\left[u_f - L_{21} \gamma_1^*\right] =  Q_2 ^{\top}L_{22}^{-1} \left[u_f - L_{21} \gamma_1^*\right]  = Q_2 ^{\top}\gamma_2^*.$$}
	\end{proof}
	From Lemma~\ref{Lemma:4} and~\ref{Lemma:5}, we can then characterize the minimum-norm parameter $\alpha$ of the whole behavioral model in \eqref{eq:prediction_model_DD1}-\eqref{eq:prediction_model_DD2} as follows.
	\begin{theorem}[Decomposition of $\alpha^\star$]\label{thm:1}
		Let $\alpha^\star_{init} \in \mathbb{R}^{N}$ and $\alpha^\star_{f} \in \mathbb{R}^{N}$ be defined as in Lemma~\ref{Lemma:4} and \ref{Lemma:5}, respectively. Then, they satisfy the following properties:
		\begin{enumerate}
			\item $\alpha^\star_{init}=Q_{1}^\top\gamma^\star_{1};$		
			\item $\alpha^\star_{f}=Q_{2}^\top \gamma^\star_{2}$;
			\item $\alpha^\star_{init}$ is orthogonal to $\alpha^\star_{f}$;
			\item $Q_{1}\alpha^\star_{f}=0$  and $Q_{2}\alpha^\star_{init}=0$ 
			\item $Q_{3}\alpha^\star_{f}=Q_{3}\alpha^\star_{init}=0$.	
		\end{enumerate}
		Therefore, $\alpha^\star= \alpha^\star_{init}+\alpha^\star_{f}$ is the minimum-norm vector satisfying the conditions:
		\begin{subequations}\label{eq:new_predictor}
			\begin{align}
			\nonumber&\begin{bmatrix}
			z_{init}\\
			u_{f}
			\end{bmatrix}\!=\!\begin{bmatrix}
			Z_{P}\\				
			U_{F}
			\end{bmatrix}\alpha^{\star}\!\!=\!\begin{bmatrix}
			L_{11} &  0\\
			L_{21} & L_{22} 				
			\end{bmatrix}\begin{bmatrix}
			Q_{1}\\
			Q_{2}
			\end{bmatrix}\alpha^{\star}\\
			&\qquad \qquad \qquad \qquad =\!\begin{bmatrix}
			L_{11} &  0\\
			L_{21} & L_{22} 				
			\end{bmatrix}\begin{bmatrix}
			\gamma^{\star}_{1}\\
			\gamma^{\star}_{2}
			\end{bmatrix},\\
			&\hat y^d_{f}=\sum_{i=1}^{3}L_{3i}Q_{i}\alpha^*=\sum_{i=1}^{2}L_{3i}Q_{i}\alpha^*=\sum_{i=1}^{2}L_{3i}\gamma^*_{i}. \label{eq:pred_output}
			\end{align}
		\end{subequations}
	\end{theorem} 
	\begin{proof}
		{Conditions 1 and 2 have  been proved in Lemmas \ref{Lemma:4} and \ref{Lemma:5} respectively. Condition 3 and 4 are  direct consequences of the fact that $Q_1^\top Q_2 = 0$.  Finally, Condition 5 derives from the fact that $Q_3^\top Q_i = 0$, for $i=1,2$. It is also straightforward to verify that, indeed, $\alpha^\star= \alpha^\star_{init}+\alpha^\star_{f}$ is a solution to \eqref{eq:new_predictor}. The fact that it is the minimum-norm solution derives from the fact that $\alpha^\star$ belongs to the column space of $[Q_1^\top \; Q_2^\top]$. }
	\end{proof}
	
	The properties highlighted above allows us to reformulate the DDPC problem as follows:
	\begin{subequations}\label{eq:DDPC_prob2}
		\begin{align}
			&\underset{{\gamma_{1},\gamma_{2}}}{\mbox{minimze}}~~J\left(\begin{bmatrix}
				y_{f}^{d}\\
				u_{f}
			\end{bmatrix}\right)\\
			&~~\mbox{s.t.}~~\begin{bmatrix}
				z_{init}\\
				u_{f}\\
				y_{f}^{d}
			\end{bmatrix}=\begin{bmatrix}
				L_{11} & 0 \\
				L_{21} & L_{22} \\
				L_{31} & L_{32} 
			\end{bmatrix}\begin{bmatrix}
				\gamma_{1}\\\gamma_{2}
			\end{bmatrix} \label{eq:prediction_model2},\\
			&\qquad~~~u(k) \in \mathcal{U},~y^{d}(k) \in \mathcal{Y},~k \in [t,t+T), \label{eq:value_constr}
		\end{align}
	\end{subequations}
	where the cost is defined in \eqref{eq:cost_DD} and we reshape the predictor based on the properties of the minimum-norm $\alpha$ highlighted in Theorem~\ref{thm:1}.
	
	By looking at \eqref{eq:DDPC_prob2}, it can be easily noticed that the cost and the value constraints in \eqref{eq:value_constr} are independent of $\gamma_{1}$. In turn, $\gamma_1$ is  solely determined by the initial conditions $z_{init}$. As such, {$\gamma_1$ is not a proper optimization variable, but acts as a constraint that can be explicitely solved by setting: 
	\begin{equation}\label{eq:init_terms}
			\gamma_{1}^\star=
			L_{11}^{-1}z_{init}.
	\end{equation}}
\rev{By leveraging LQ-decomposition, the problem of matching initial conditions can thus be decoupled from that of designing the optimal input.} It is worth stressing once more that, according to Lemma~\ref{Lemma:4}, $\gamma^\star_{1}$ and $\gamma_{2}$ found through \eqref{eq:init_terms} coincides with the ones leading to the minimum-norm $\alpha$ satisfying the initial conditions.
	
	The constrained optimization problem to be solved at each time instant thus result into a reduced problem on $\gamma_2$ only, i.e.,
		\begin{subequations}\label{eq:DDPC_prob3}
		\begin{align}
			&\underset{{\gamma_{2}}}{\mbox{minimize}}~~J\left(\begin{bmatrix}
				y_{f}^{d}\\
				u_{f}
			\end{bmatrix}\right)\\
			&~~\mbox{s.t.}~~\begin{bmatrix}
				u_{f}\\
				y_{f}
			\end{bmatrix}=\begin{bmatrix}
				L_{21} & L_{22} \\
				L_{31} & L_{32}
			\end{bmatrix}\begin{bmatrix}
				\gamma_{1}^\star\\\gamma_{2}
			\end{bmatrix} \label{eq:prediction_model3},\\
			&~~~~~~~~~u(k) \in \mathcal{U},~y^{d}(k) \in \mathcal{Y},~k \in [t,t+T), \label{eq:constraints2}
		\end{align}
	\end{subequations}
	with $\gamma_1(t)$ fixed at the solution of \eqref{eq:init_terms}.
	
		\begin{algorithm}[!tb]
		\caption{$\gamma$-DDPC at time $t$}
		\label{algo1}
		~\textbf{Input}: Matrices $\{L_{i,j}\}_{i=1}^{3}$, $j=1,2$; penalties $Q \succeq 0$, $R \succ 0$; target $y_{r}$; constraint sets $\mathcal{U}$ and $\mathcal{Y}$; initial conditions $z_{init}$.
		\vspace*{.1cm}\hrule\vspace*{.1cm}
		\begin{enumerate}[label=\arabic*., ref=\theenumi{}]
			\item\label{step:1} \textbf{Find} $\gamma_{1}^\star$ via \eqref{eq:init_terms}; 
			\item \textbf{Optimize} $\gamma_{2}$ by solving \eqref{eq:DDPC_prob3};
			\item \textbf{Construct} $u_{f}$ according to \eqref{eq:stepalgo};
			\item \textbf{Extract} the first optimal input from $u_{f}$.
		\end{enumerate}
		\vspace*{.1cm}\hrule\vspace*{.1cm}
		~\textbf{Output}: Optimal input $u^{\star}(t)$.
	\end{algorithm}

	According to this decomposition, we propose the $\gamma$-DDPC scheme, summarized in Algorithm~\ref{algo1}. Apart from inheriting the properties of the predictor highlighted using the LQ decomposition with respect to noise handling, the $\gamma$-DDPC scheme is likely to be computationally advantageous. Indeed, the dimension of the optimization variable $\gamma_{2} \in \mathbb{R}^{mT}$ in \eqref{eq:DDPC_prob3} is likely to be considerably smaller than the one of $\alpha \in \mathbb{R}^{N}$. At the same time, retrieving $\gamma_{1}$ requires the inversion of a matrix with dimensions dictated by the chosen $\rho$. 
	
	\begin{remark}[Choice of $\rho$ (part II)]\label{rem:rho2}
		The length of the ``past'' window plays a pivotal role in shaping the performance of the predictive controller. On the one hand, $\rho$ should be chosen by following an identification-oriented reasoning (see Remark~\ref{rem:rho}). On the other, a smaller $\rho$ reduces the number of data needed to solve the DDPC problem ($N_{data}\!:=\!N\!+\!T\!+\!\rho$), and it would result in a computationally lighter DDPC problem. Its value has thus to be selected by trading-off between these  requirements. 
	\end{remark} 

	\subsection{Explaining regularization in DDPC}\label{Sec:explain}
	By looking at the DDPC problem from a different angle, the results presented so far allow us to have a clearer vision on the actual effect that additional regularization terms have on the optimal control action generated when solving \eqref{eq:DD_RHPC_prob_mean}. \rev{We stress that the use of regularization is currently, by and large,  the strategy proposed by most of the literature to cope with stochastic noise in DDPC. }
	
	The properties highlighted in Theorem~\ref{thm:1} indicate that $Q_{3}\alpha$ should be set to zero, if one seeks to reduce the effect of noise on the predictions exploited to determine the optimal control action. At the same time, one should not excessively shrink the values of $Q_{i}\alpha$, for $i=1,2$. While these two conflicting requirements on $\alpha$ can be easily accommodated when decomposing the predictor using the LQ decomposition, this operation is not as easy when the predictor in \eqref{eq:prediction_model_DD1}-\eqref{eq:prediction_model_DD2} is used as it is. Indeed, in this last case, one can only try shrink the whole vector $\alpha$ by introducing a regularizer in the cost, as already proposed in \cite{dorfler2022bridging,berberich2020data}. Although such procedure has proven to be effective, the regularization strength has to be well calibrated to trade-off between reducing the norm of $\alpha$ and retaining the information needed to produce a meaningful control action. In turn, achieving this balance requires the fine tuning of the regularization penalty, representing a well-known drawback of regularization-based DDPC approaches. Indeed, existing procedures generally require closed-loop experiments to calibrate the regularization parameters, which can endanger the safety of the plant, ultimately limiting the applicability of existing DDPC strategies.

	\section{A benchmark case study}\label{sec:benchmark}
	\begin{figure}[!tb]
		\centering
		\begin{tabular}{cc}
					\hspace{-1cm}
					\subfigure[$\mathcal{J}$ \emph{vs} $\overline{\mbox{SNR}}$\label{Fig:SNRJ}]{\includegraphics[scale=.4,trim=2cm 1.75cm 5cm 17cm,clip]{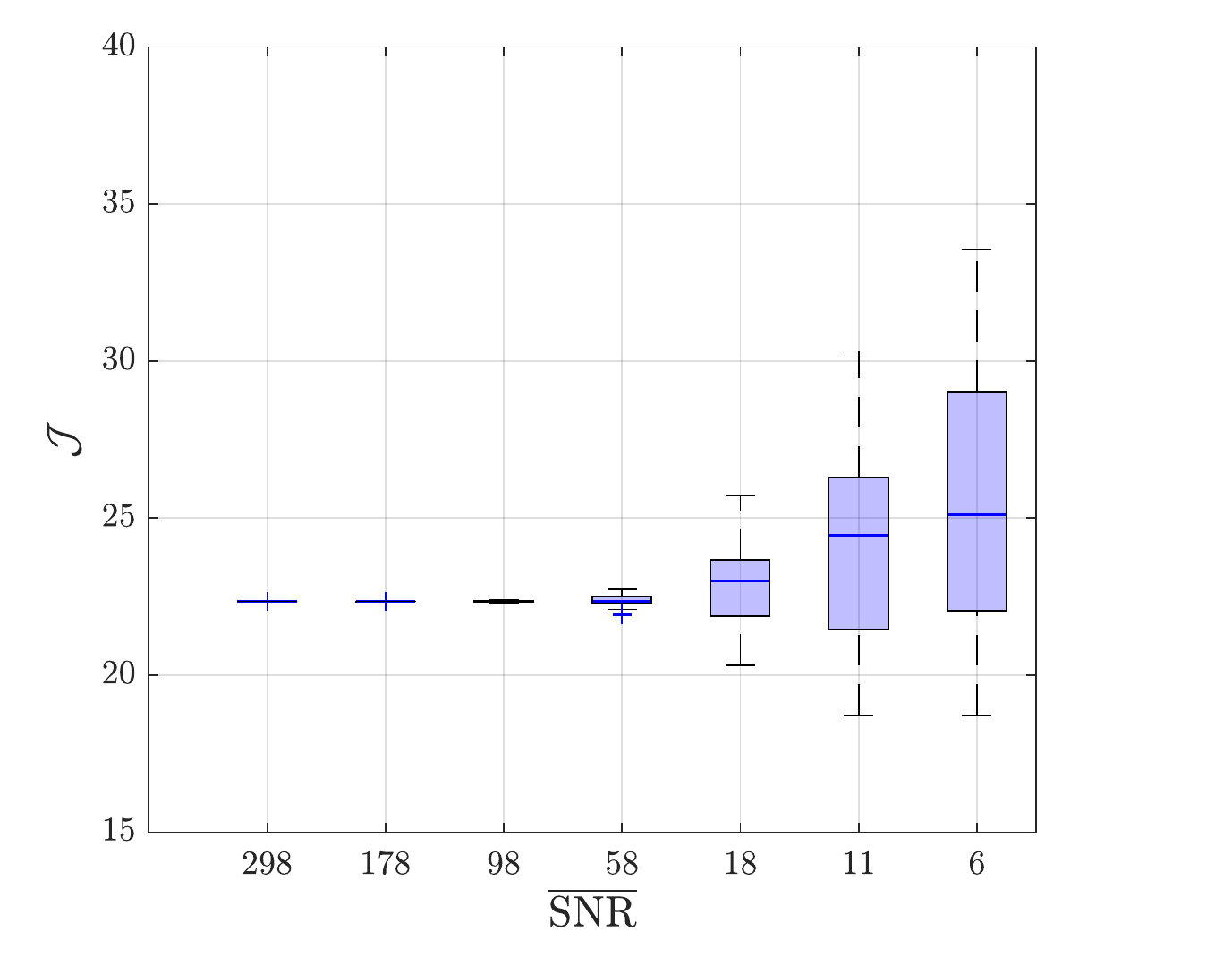}} \hspace{-.7cm} & \hspace{-.7cm} 
				\subfigure[$\mathcal{J}_u$ \emph{vs} $\overline{\mbox{SNR}}$\label{Fig:SNRJu}]{\includegraphics[scale=.4,trim=2cm 1.75cm 5cm 17cm,clip]{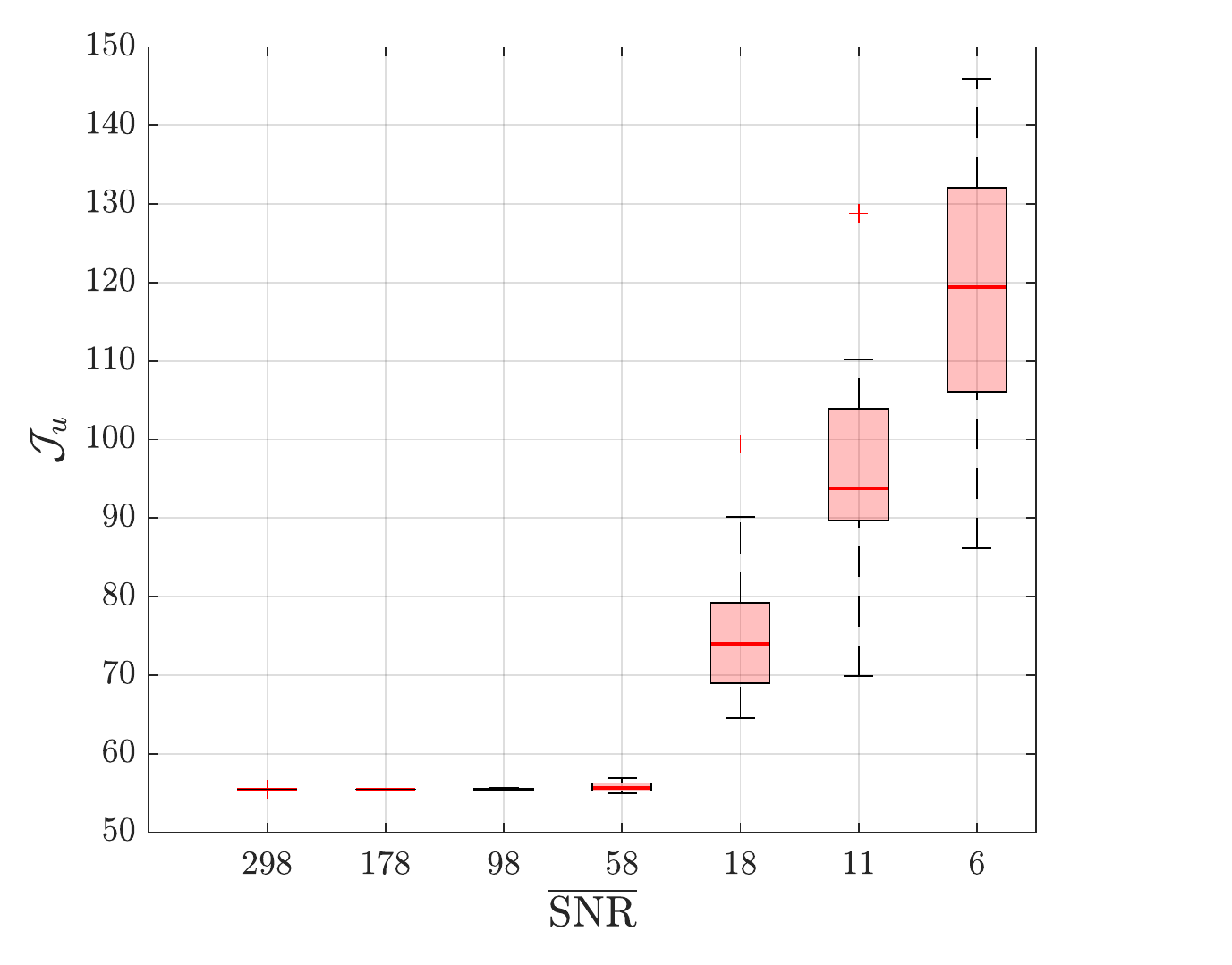}}
		\end{tabular}
		\caption{\rev{Closed-loop validation tests: performance indexes \emph{vs} average \emph{signal-to-noise ratio} ($\overline{\mbox{SNR}}$) over $30$ Monte Carlo predictors}.}\label{Fig:different_noiseLevels} \vspace{-.7cm}
	\end{figure}
	\begin{figure}[!tb]
		\centering
		\begin{tabular}{cc}
		\hspace*{-.4cm}\includegraphics[scale=.515,trim=2cm 2cm 10cm 20cm,clip]{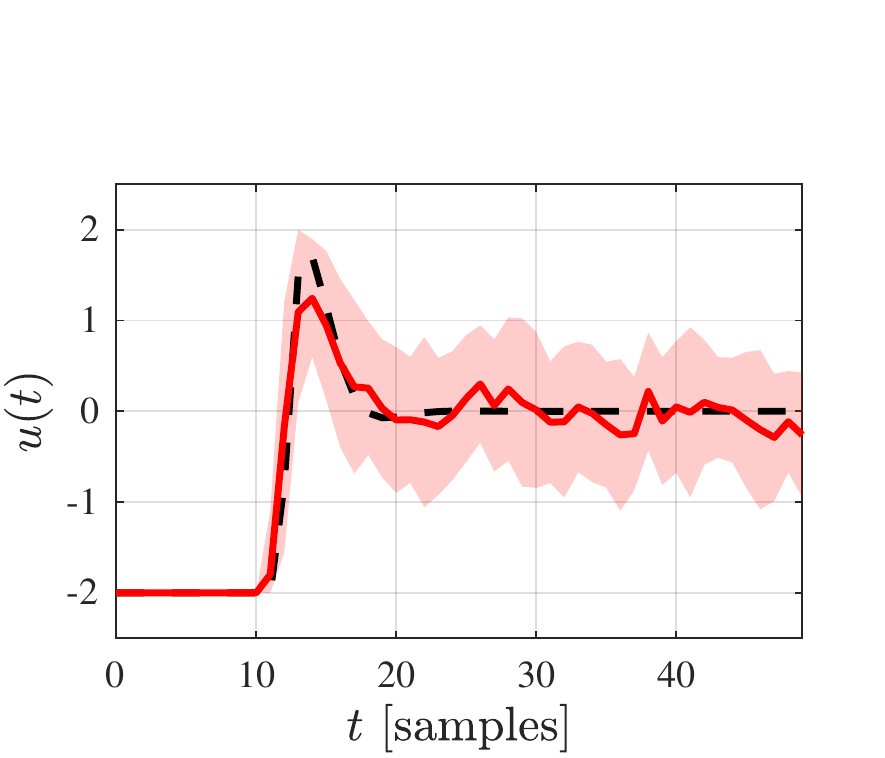} \hspace*{-.4cm}&\hspace*{-.4cm}	\includegraphics[scale=.515,trim=2cm 2cm 10cm 20cm,clip]{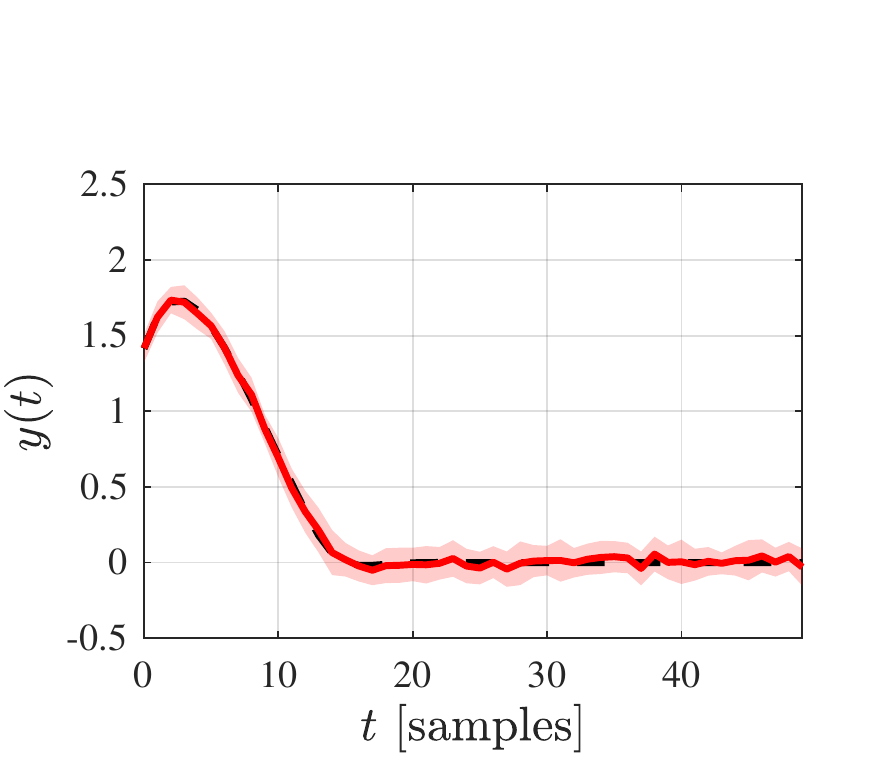}
		\end{tabular}
		\caption{\rev{Closed-loop validation tests ($\overline{\mbox{SNR}}=18$~dB): average (red) closed-loop response and inputs with their standard deviations (shaded area) over the $30$ predictors \emph{vs} oracle \emph{noise-free} MPC (black dashed lines).}}\label{Fig:diff_lambda1noisy}
		\end{figure}
		\begin{figure}[!tb]
		\centering
		\begin{tabular}{c}
			\subfigure[$\mathcal{J}$ \emph{vs} predictive strategy\label{Fig:PredJ}]{\includegraphics[scale=.4,trim=3.05cm 1.95cm 0cm 17cm,clip]{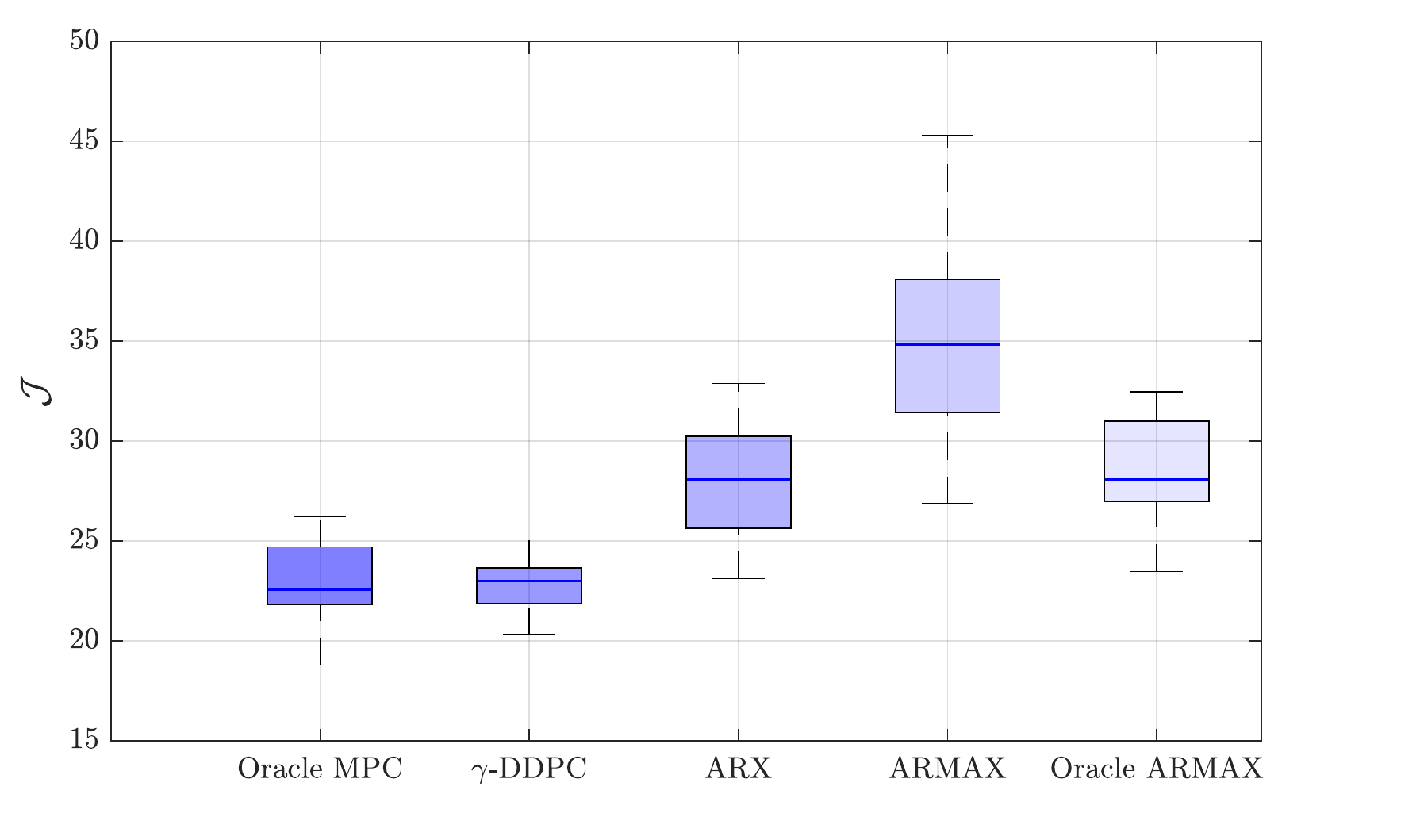}} \vspace{-.2cm}\\
			\subfigure[$\mathcal{J}_u$ \emph{vs} predictive strategy\label{Fig:PredJu}]{\includegraphics[scale=.4,trim=3.05cm 1.95cm 0cm 17cm,clip]{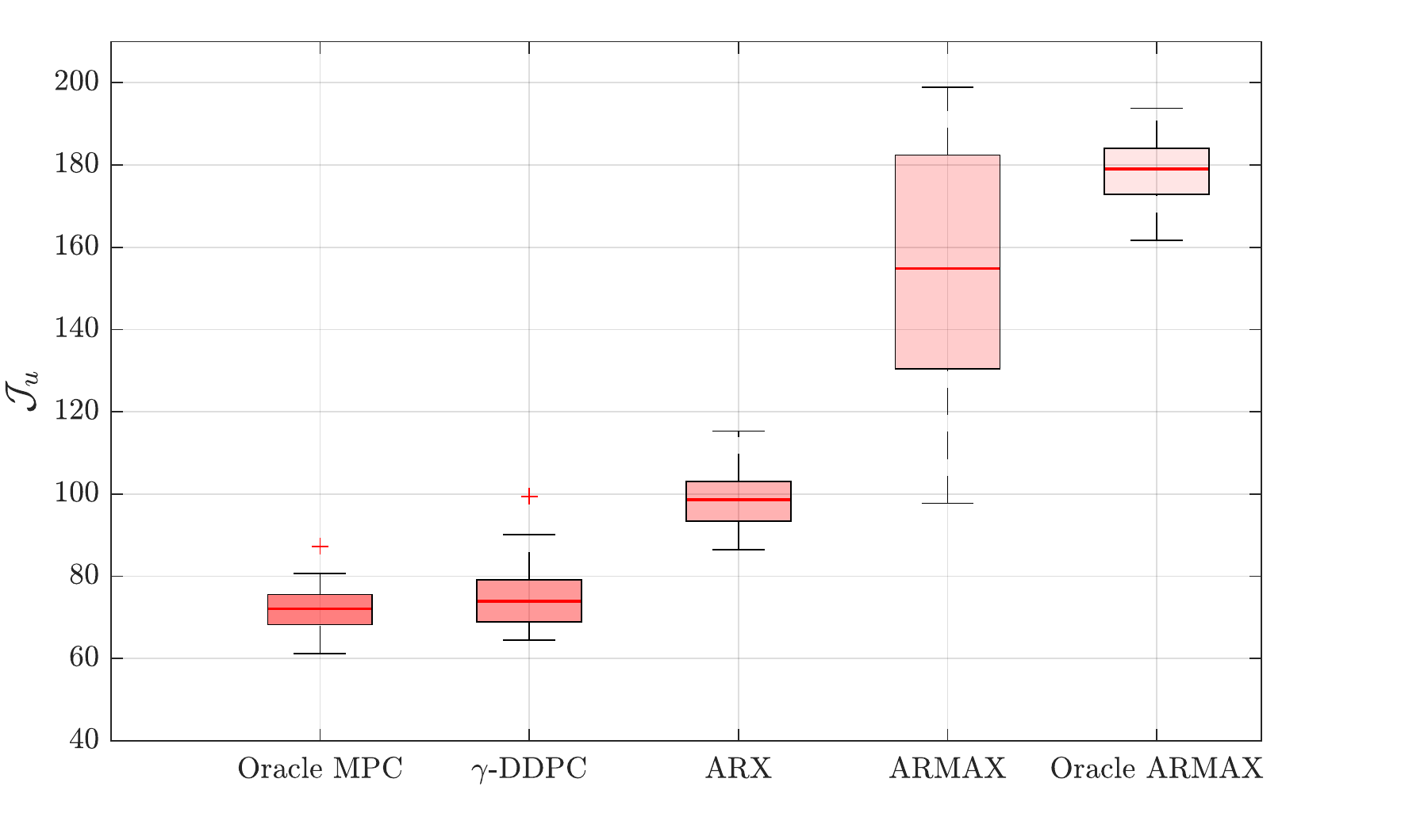}}
		\end{tabular}
		\caption{\rev{Closed-loop validation tests ($\overline{\mbox{SNR}}=18$~dB): performance indexes \emph{vs} predictive strategy over $30$ Monte Carlo predictors}.}\label{Fig:different_predictors}
	\end{figure}
		\begin{figure}[!tb]
	\centering
	\begin{tabular}{c}
		\subfigure[$|\mathcal{J}-\bar{\mathcal{J}}^{\mathrm{o}}|$ \emph{vs} $N_{data}$\label{Fig:DataJ}]{\includegraphics[scale=.4,trim=2cm 1.75cm 5cm 17cm,clip]{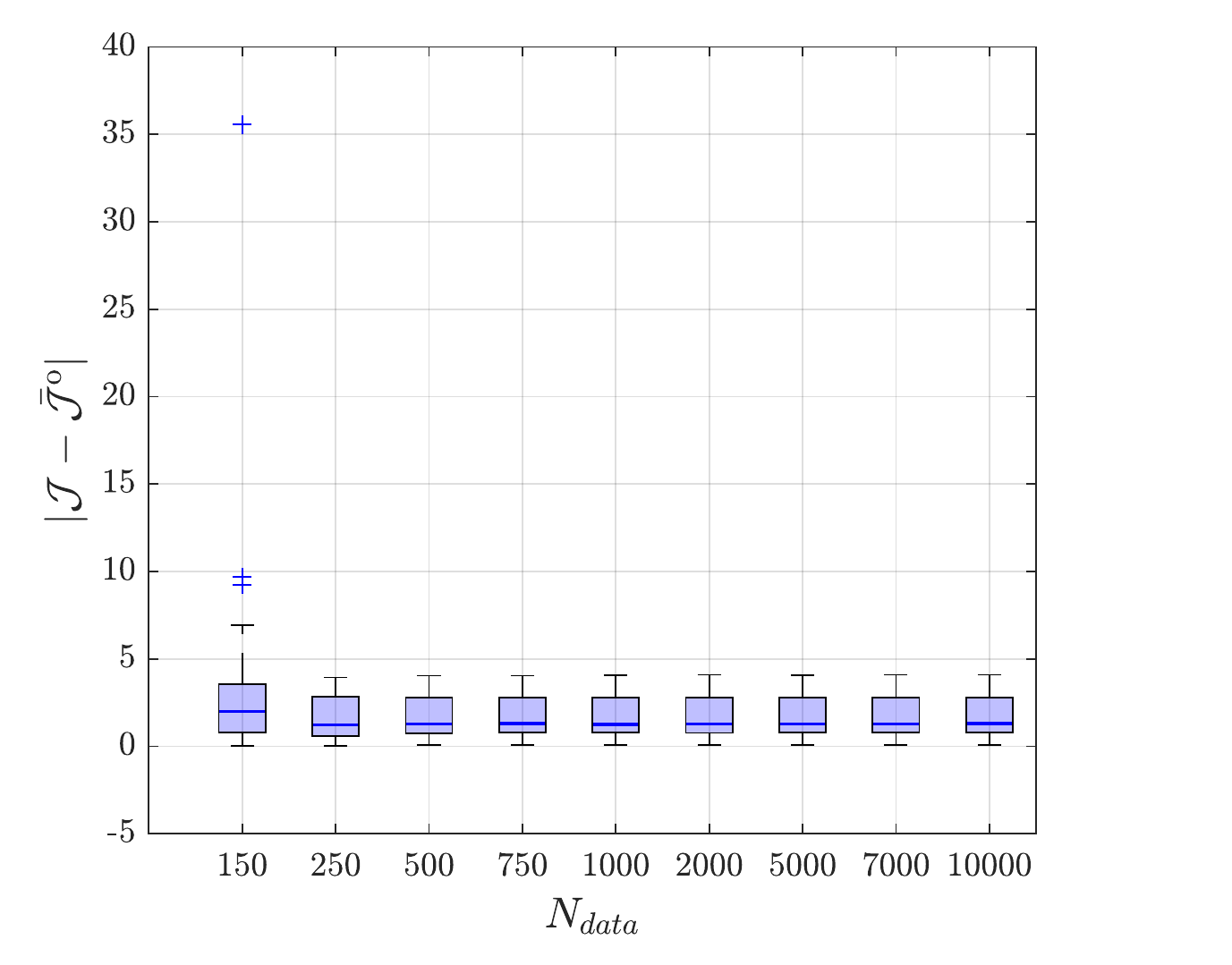}} \\
		\subfigure[$|\mathcal{J}_u-\bar{\mathcal{J}}_{u}^{\mathrm{o}}|$ \emph{vs} $N_{data}$\label{Fig:DataJu}]{\includegraphics[scale=.4,trim=2cm 1.75cm 5cm 17cm,clip]{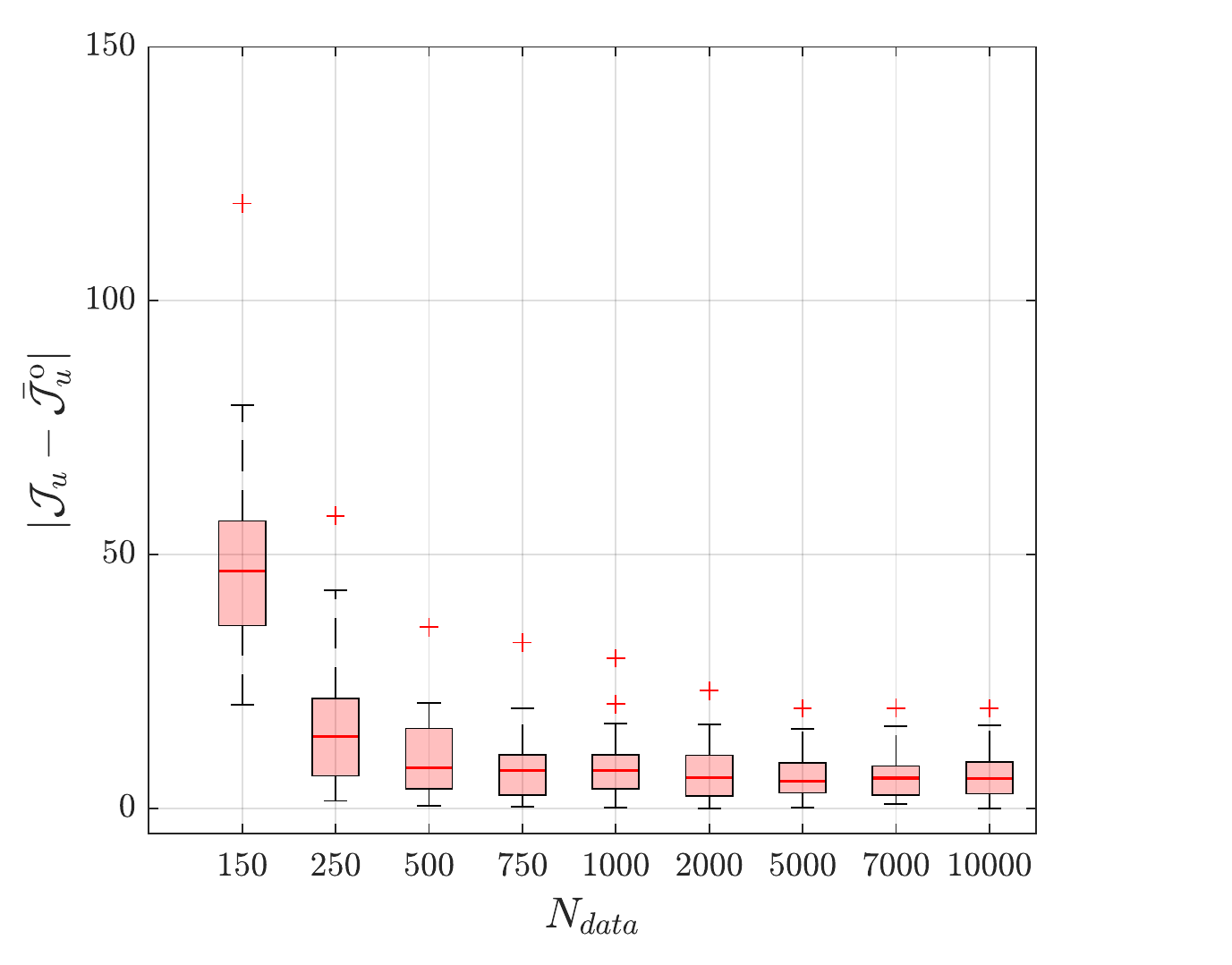}}
	\end{tabular}
	\caption{\rev{Closed-loop validation tests ($\overline{\mbox{SNR}}=18$~dB): absolute differences between the performance indexes of $\gamma$-DDPC and the average values of those associated with the \emph{noisy} oracle MPC \emph{vs} $N_{data}$ over $30$ Monte Carlo predictors}.}\label{Fig:data}
\end{figure}
\begin{figure}[!tb]
	\centering
	\begin{tabular}{cc}
		\hspace{-1cm}
		\subfigure[$|\mathcal{J}-\bar{\mathcal{J}}^{\mathrm{o}}|$ \emph{vs} $\rho$\label{Fig:rhoJ}]{\includegraphics[scale=.4,trim=2cm 1.75cm 5cm 17cm,clip]{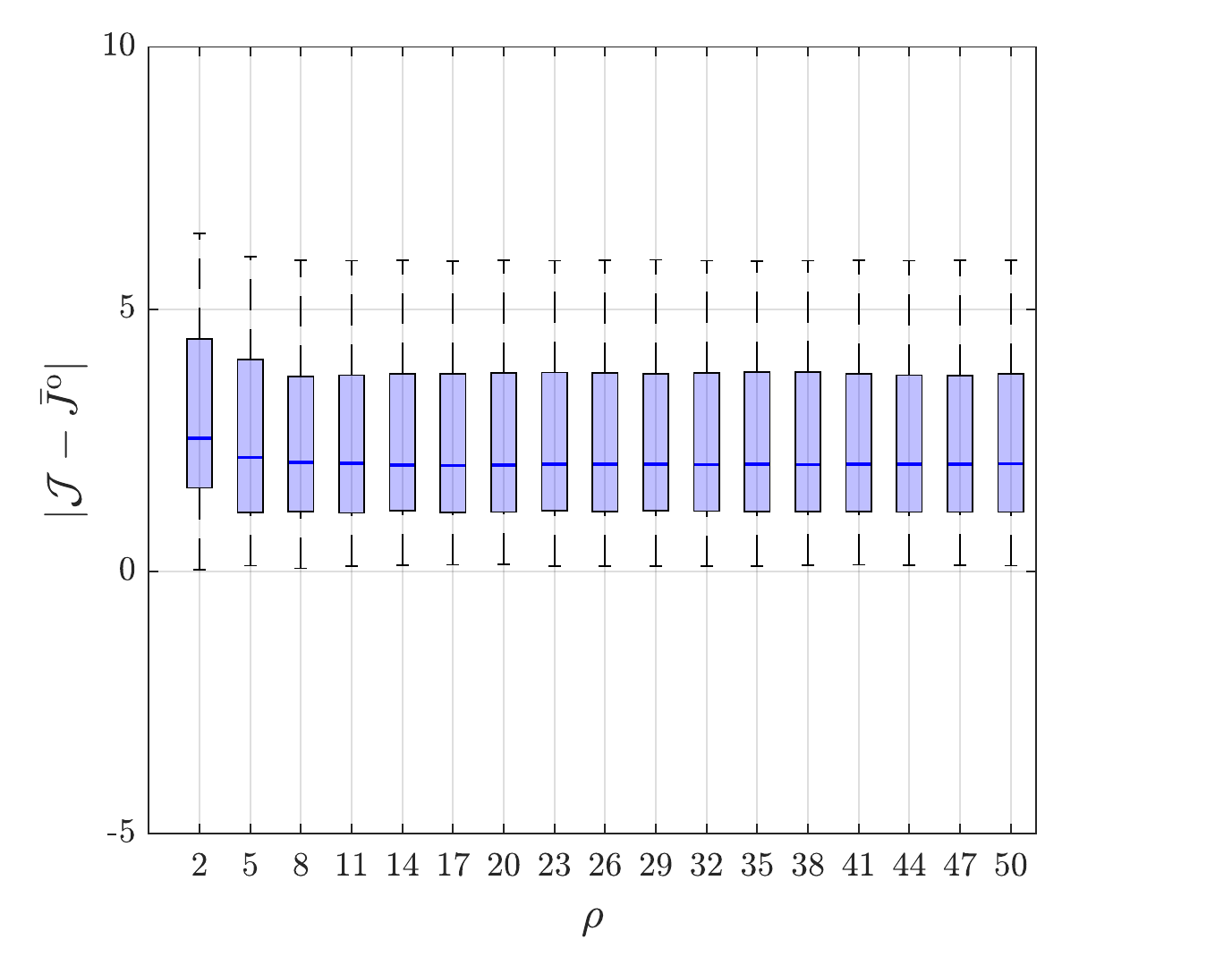}} \hspace{-.7cm} &\hspace{-.7cm} 
		\subfigure[$|\mathcal{J}_u-\bar{\mathcal{J}}_{u}^{\mathrm{o}}|$ \emph{vs} $\rho$\label{Fig:rhoJu}]{\includegraphics[scale=.4,trim=2cm 1.75cm 5cm 17cm,clip]{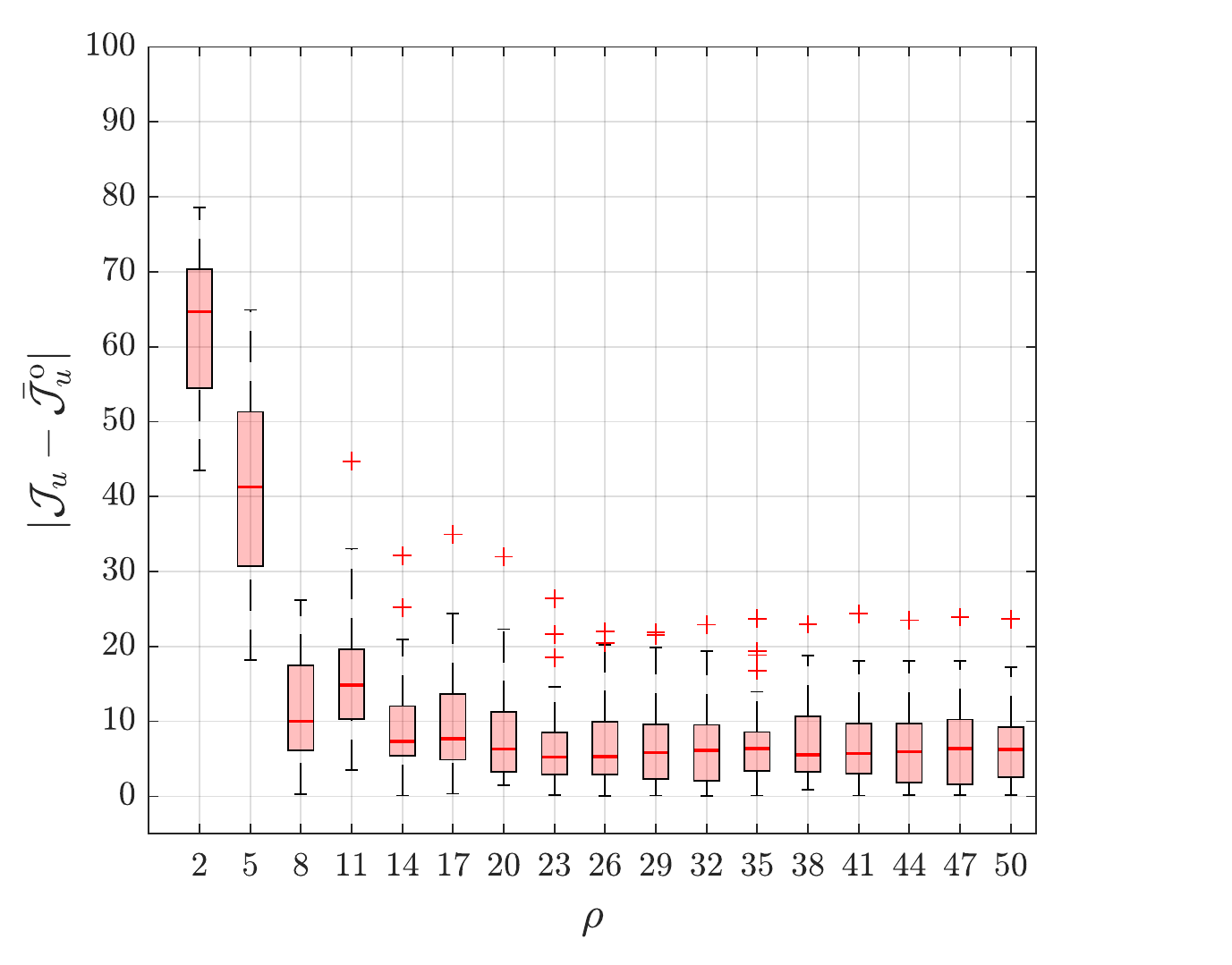}}
	\end{tabular}
	\caption{\rev{Closed-loop validation tests ($\overline{\mbox{SNR}}=18$~dB): absolute differences between the performance indexes of $\gamma$-DDPC and the average values of those associated with the \emph{noisy} oracle MPC \emph{vs} length of the \textquotedblleft past horizon\textquotedblright \ $\rho$ over $30$ Monte Carlo predictors}.}\label{Fig:different_rho}
\end{figure}
To assess the effectiveness of the proposed $\gamma$-DDPC scheme, while validating the conclusions drawn in Section~\ref{Sec:explain}, we consider the same benchmark example proposed in \cite{bemporad2002explicit}. Therefore, the unknown plant to be controlled is described by the following model:
	\begin{equation}\label{eq:model_ex}
		\begin{cases}
			x(t+1)\!=\!\begin{bmatrix}
			0.7326 & -0.0861\\
			0.1722 & 0.9909
			\end{bmatrix}\!x(t)\!+\!\!\begin{bmatrix}
			0.0609\\
			0.0064
		\end{bmatrix}\!u(t)\!+\!Ke(t),\\
			y(t)\!=\!\begin{bmatrix} 0 & 1.4142 
				\end{bmatrix}\!x(t)\!+\!e(t),
		\end{cases}
	\end{equation}
where the innovation is set to be zero-mean and Gaussian distributed\rev{, while $K \in \mathbb{R}^{1 \times 2}$ is randomly chosen according to a normal distribution, with all the eigenvalues of $A-KC$ being inside the unit circle.} By considering a prediction horizon of length $T=40$ and $\rho=23$ (selected according to Remark~\ref{rem:rho}), we design the predictive controllers solving a zero regulation problem by running Algorithm~\ref{algo1}\footnote{All tests have been carried out on an M1 chip, running MATLAB 2021a, while the optimization problems are solved with CVX \cite{gb08,cvx}} with $Q=I$, $R=10^{-3}$ and \rev{$y_{r}=u_r=0$}
, as in \cite{bemporad2002explicit}. To have a quantitative assessment of performance, for all closed-loop tests we consider the following indexes:
\begin{subequations} \label{eq:perf_indexes}
	\begin{align*}
	& \mathcal{J}=\sum_{t=0}^{T_{v}-1} \|y(t)\|_{Q}^{2}+\|u(t)\|_{R}^{2},~~~ \mathcal{J}_{u}=\sum_{t=0}^{T_{v}-1} u(t)^{2},
	\end{align*}
\end{subequations}  
that allow us to have a compact information on the tracking performance and the input effort in testing.

\rev{We initially focus on assessing the performance of the $\gamma$-DDPC scheme introduced in Section~\ref{Sec:new_scheme}. Firstly, we assess the sensitivity of $\gamma$-DDPC to noise in the available batch of data $\mathcal{D}_{N_{data}}$. By progressively increasing the level of noise, we thus perform $30$ Monte Carlo simulations of length $N_{data}=1000$ with a random input sequence, uniformly distributed in the interval $[-5,5]$, to generate different datasets. Closed-loop performance is then evaluated for each predictive model and level of noise by using $\gamma$-DDPC to close the control loop over tests of length $T_{v}=50$, always starting from the same initial condition. As shown in \figurename{~\ref{Fig:different_noiseLevels}}, b}oth the performance and control effort are quite consistent when the average \emph{signal-to-noise ratio} (SNR) is high. Instead, a slight degradation in performance is experienced when the average noise corrupting the data used to construct the predictor decreases, along with an increase in the control effort required during \rev{closed-loop} testing. These results generally show that the proposed $\gamma$-DDPC strategy \rev{allows the closed-loop system to track (on average) the reference, in spite of the process and measurement noise affecting it}. This consideration is further confirmed by the \rev{results reported} in \figurename{~\ref{Fig:diff_lambda1noisy}}, \rev{where the closed-loop inputs and output attained with $\gamma$-DDPC are compared with the ones obtained via an MPC designed with the true system matrices (denominated from now on \emph{oracle MPC}) within a noise-free setting.}

\rev{For a fixed level of noise, we then} compare the closed-loop performance achieved with Algorithm~\ref{algo1} with the ones attained by designing an MPC with an identified model\footnote{The model is identified with N4SID \cite{N4SID}.} of the plant. To this end, we keep the input/output structure of the predictor by identifying both an \emph{autoregressive} model \emph{with exogenous inputs} (ARX) of order $23$, an ``oracle'' \emph{autoregressive moving average} model \emph{with exogenous inputs} (ARMAX) of order $2$ and ARMAX models with orders selected according to Remark~\ref{rem:rho}\footnote{The average order of the ARMAX models is $7$, while its standard deviation is $5$.}. As shown in \figurename{~\ref{Fig:different_predictors}}, the use of all identified models tends to slightly deteriorate performance, while requiring an additional control effort. \rev{The main deterioration in performance is visible when} the ARMAX models are used to design the MPC. 
These results thus highlight the possible advantages of using the $\gamma$-DDPC scheme over a identification+model-based control procedure, at least for the considered case study. We also evaluate how $\gamma$-DDPC performs when increasing $N_{data}$ over noisy closed-loop tests. As shown in \figurename{~\ref{Fig:data}}, the difference between the overall cost and the required control effort tends to decrease with the number of data, in line with established results in system identification. \rev{Lastly, we assess the sensitivity of $\gamma$-DDPC to the only free parameter of this scheme, namely $\rho$. } As shown in \figurename{~\ref{Fig:different_rho}}, \rev{the main changes due to different choices of the \textquotedblleft past horizon\textquotedblright are visible in the index assessing the control effort. In particular,} excessively small values of $\rho$ results into the demand for \rev{a greater control effort than that required by the oracle MPC. By increasing $\rho$, the input effort required by $\gamma$-DDPC tends to become aligned with that associated with the oracle MPC, while slighly increasing again when $\rho>30$. Note that $\mathcal{J}_{u}$ gets the closest to the average input effort index of the oracle MPC for $\rho=23$, thus validating the choice we have automatically performed through the Akaike's criterion.}

\subsection{Effect of additional regularization on $\gamma$-DDPC}
\begin{figure}[!tb]
	\centering
	\begin{tabular}{cc}
		\hspace{-1cm}\subfigure[$|\mathcal{J}-\bar{\mathcal{J}}^{\mathrm{o}}|$ \emph{vs} $\beta$\label{Fig:betaJ}]{\includegraphics[scale=.4,trim=2cm 1.75cm 5cm 17cm,clip]{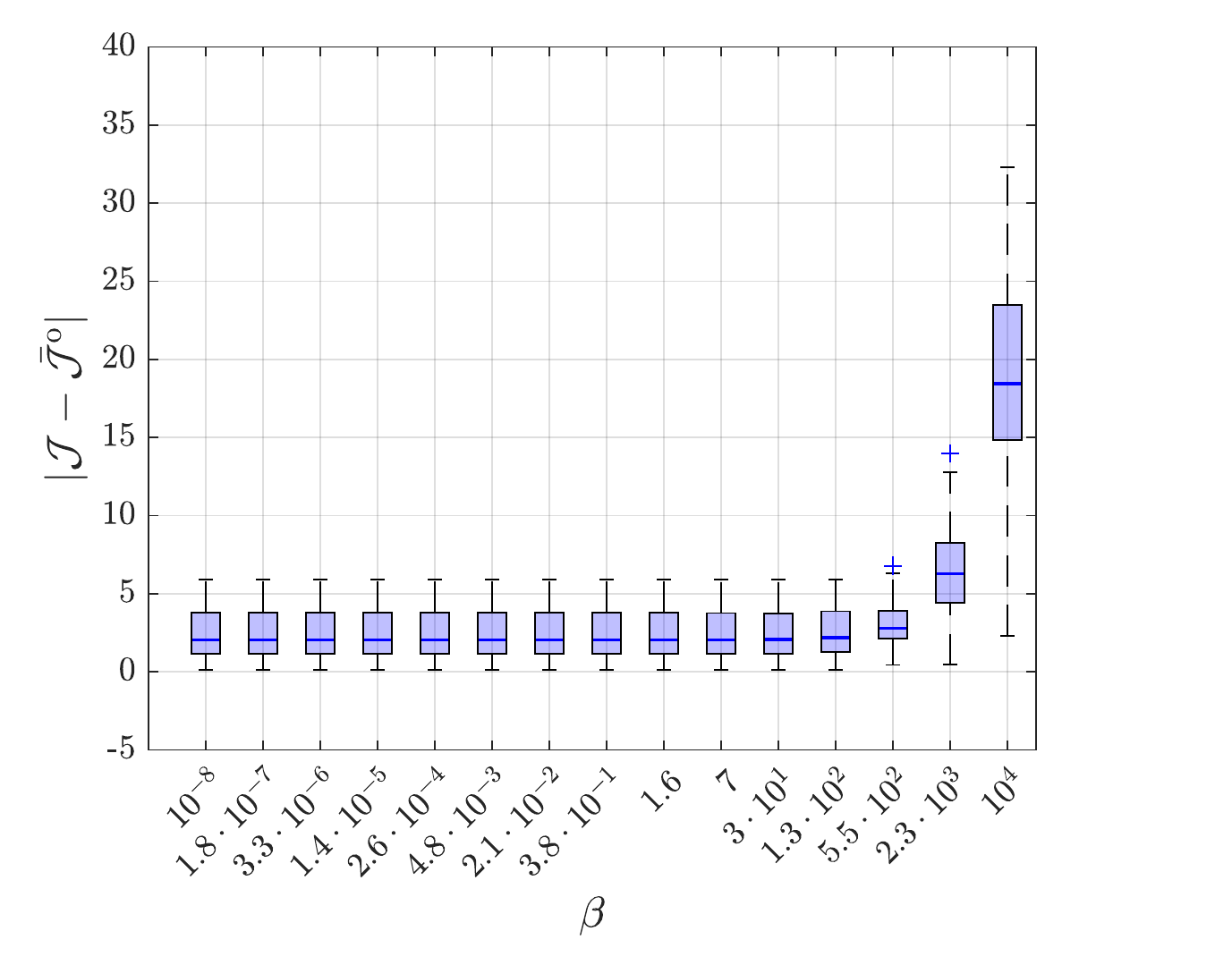}} \hspace{-.7cm} & \hspace{-.7cm}
		\subfigure[$|\mathcal{J}_u-\bar{\mathcal{J}}_{u}^{\mathrm{o}}|$ \emph{vs} $\beta$\label{Fig:betaJu}]{\includegraphics[scale=.4,trim=2cm 1.75cm 5cm 17cm,clip]{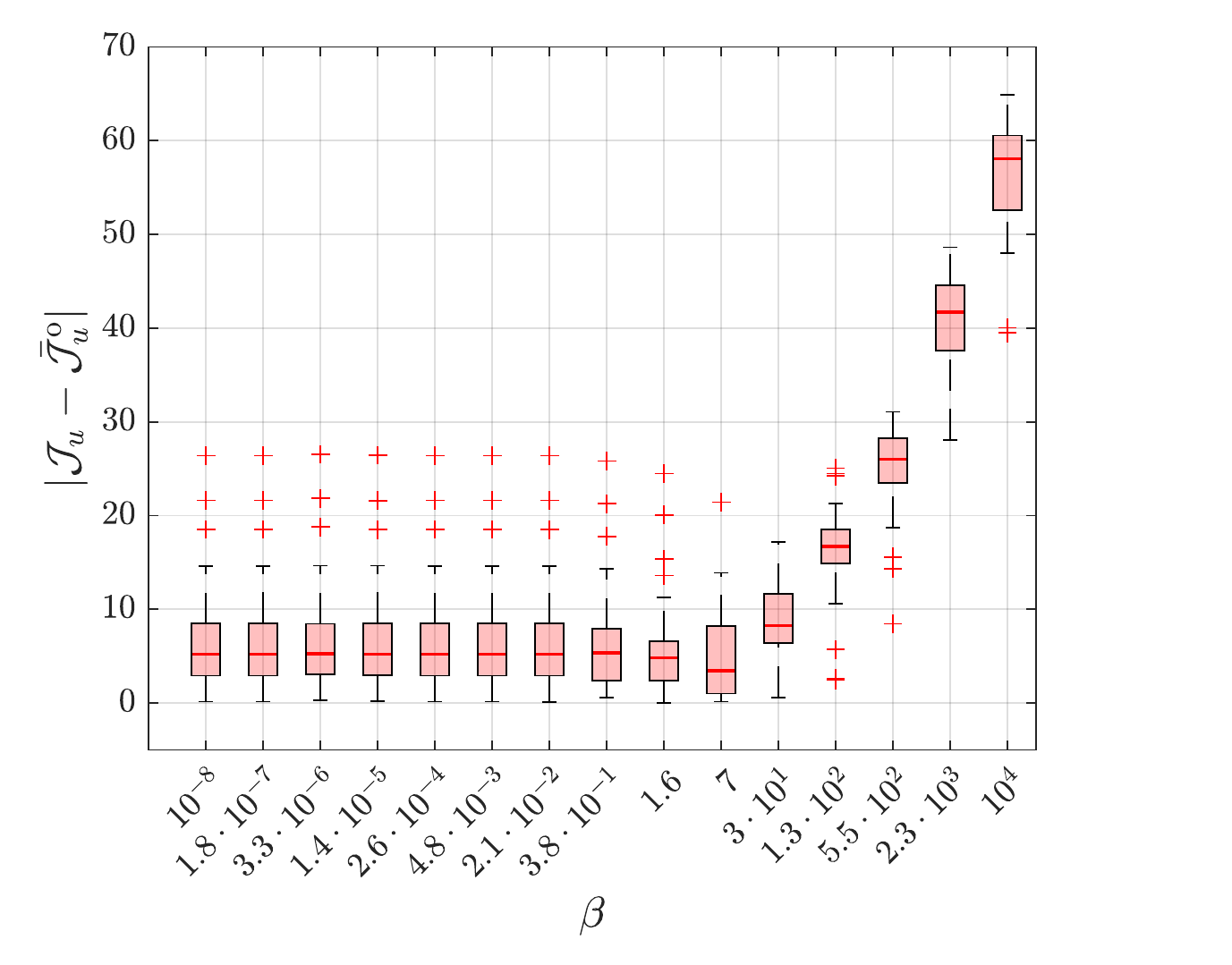}}
	\end{tabular}
	\caption{\rev{Closed-loop validation tests ($\overline{\mbox{SNR}}=18$~dB): absolute differences between the performance indexes of $\gamma$-DDPC and the average values of those associated with the \emph{noisy} oracle MPC \emph{vs} penalties $\beta$ on a 2-norm regularization on $\gamma_{2}$ over $30$ Monte Carlo predictors}.}\label{Fig:different_beta}
\end{figure}
\begin{figure}[!tb]
	\centering
	\begin{tabular}{c}
		\subfigure[$|\mathcal{J}-\bar{\mathcal{J}}^{\mathrm{o}}|$ \emph{vs} $\eta$\label{Fig:etaJ}]{\includegraphics[scale=.4,trim=2cm 1.75cm 5cm 17cm,clip]{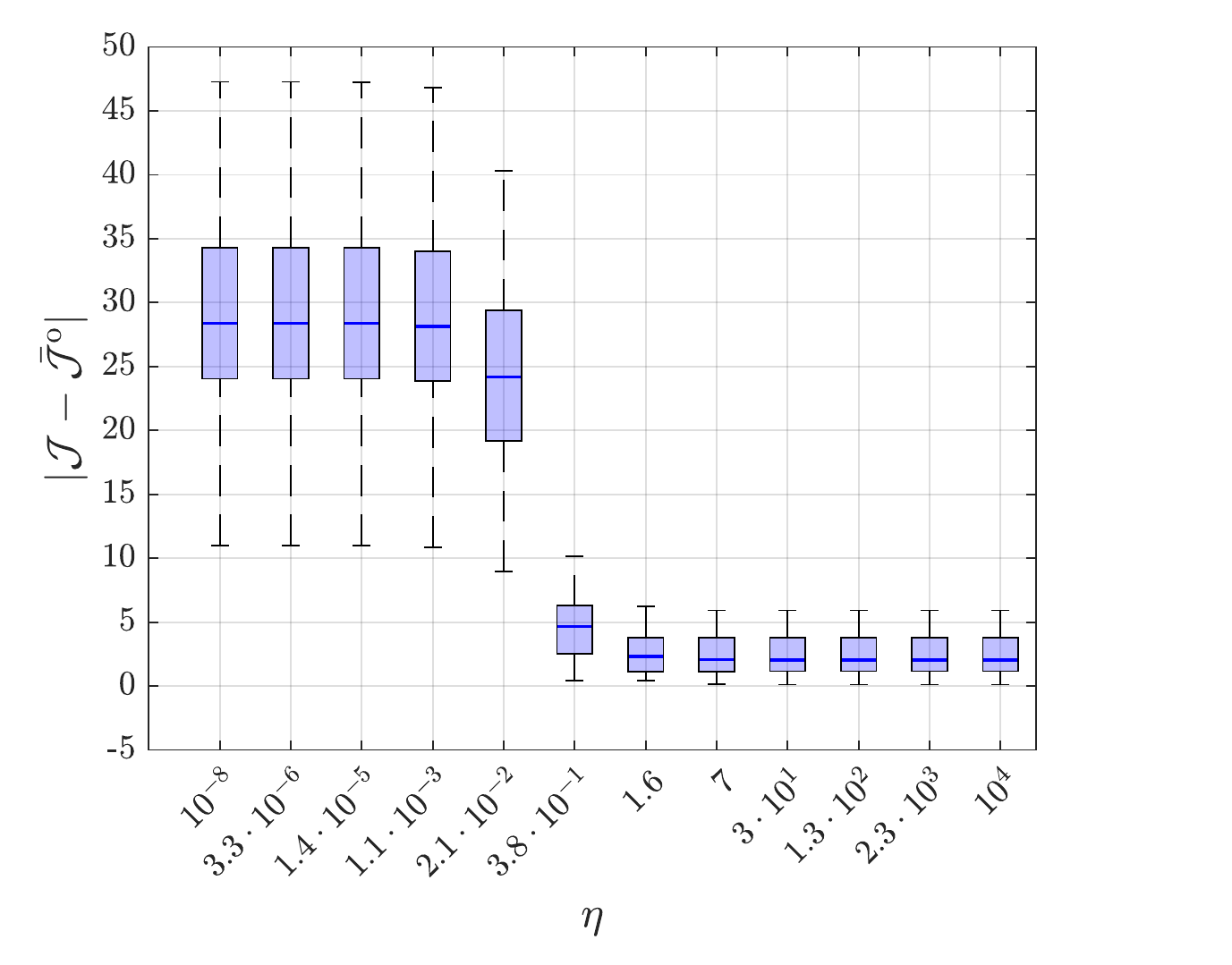}} \\
		\subfigure[$|\mathcal{J}_u-\bar{\mathcal{J}}_{u}^{\mathrm{o}}|$ \emph{vs} $\eta$\label{Fig:etaJu}]{\includegraphics[scale=.4,trim=2cm 1.75cm 5cm 17cm,clip]{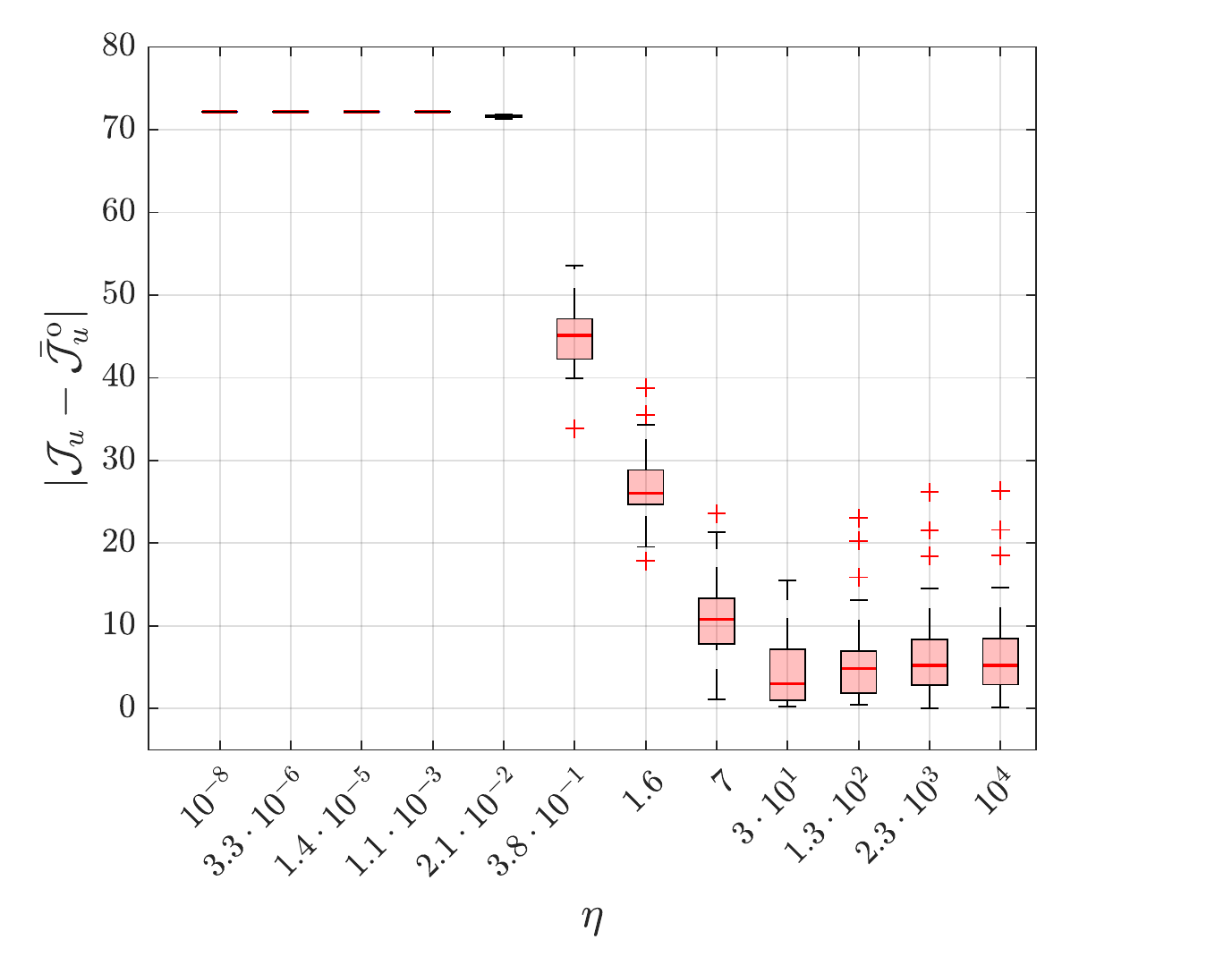}}
	\end{tabular}
	\caption{\rev{Closed-loop validation tests ($\overline{\mbox{SNR}}=18$~dB): absolute differences between the performance indexes of $\gamma$-DDPC and the average values of those associated with the \emph{noisy} oracle MPC \emph{vs} penalties $\eta$ on a 2-norm regularization on $\gamma_{3}$ over $30$ Monte Carlo predictors}.}\label{Fig:different_eta}
\end{figure}

\rev{By keeping the level of noise acting on the batch and online data, w}e now study the effect of an additional 2-norm regularization on $\gamma_{2}$, with $\beta>0$ indicating the associated penalty. As shown in \figurename{~\ref{Fig:different_beta}}, the performance index $\mathcal{J}$ tends to be rather insensitive to the additional regularization term up to a certain value of $\beta$. However, when $\beta$ increases sufficiently, performance tends to deteriorate, while the input effort tends to consistently increase with respect to the oracle MPC one. Since such a behavior is certainly undesirable, this result validates in this experimental case the claims in Section~\ref{Sec:explain}. Indeed, the additional regularization leads to a deterioration of performance, likely to be induced by the change that the regularization enforces on the actual performance-oriented cost. To prove the effectiveness of our structural choices, within the same framework we consider the DDPC problem with the predicted output defined as
\begin{equation*}
			y_{f}=\sum_{i=1}^{3}L_{3i}\gamma_{i},
\end{equation*} 
and $\gamma_{3}$ not set to zero beforehand, as in the proposed $\gamma$-DDPC approach. In this case, $\gamma_{3}$ is steered towards small values via an addition of a 2-norm regularization term in the cost \rev{weighted by $\eta>0$}. As shown in \figurename{~\ref{Fig:different_eta}}, only by heavily regularizing $\gamma_{3}$ we obtain performance comparable with the ones obtained with the oracle predictive controller. Specifically low $\eta$ result in an ineffective DDPC scheme, with the system actually operating in open-loop. These results once again show the expected detrimental effect of poor choices of the regularization parameter, highlighting the advantages of embedding insights given by subspace identification into the predictor used in the DDPC scheme. 
	\subsection{\rev{Validating results on regularized DDPC schemes}}	
	\begin{figure}[!tb]
		\centering
		\begin{tabular}{cc}
			\hspace{-1cm}\subfigure[$|\mathcal{J}-\bar{\mathcal{J}}^{\mathrm{o}}|$ \emph{vs} $\lambda$\label{Fig:luciaJ}]{\includegraphics[scale=.4,trim=2cm 1.75cm 5cm 17cm,clip]{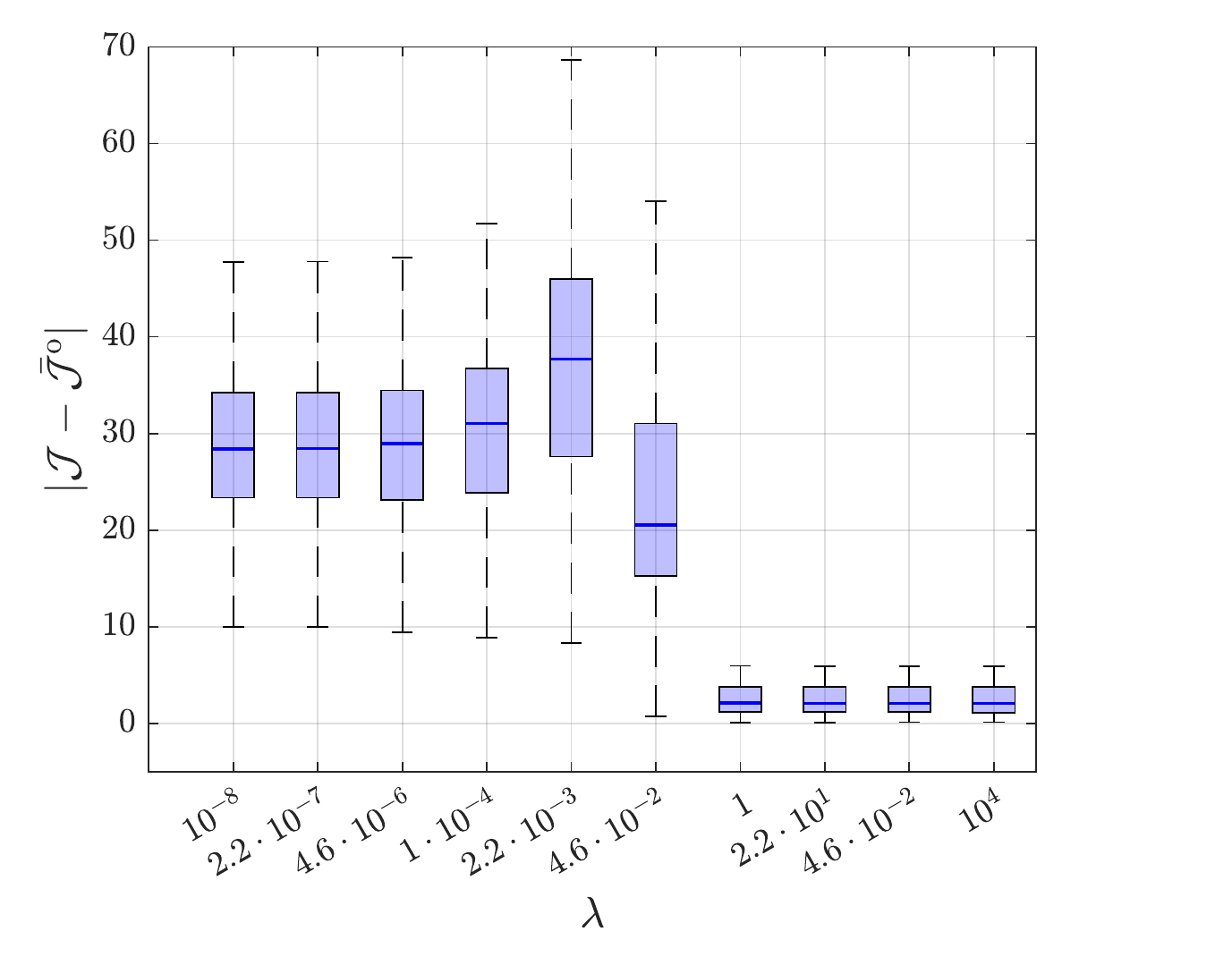}} \hspace{-.7cm} & \hspace{-.7cm}
			\subfigure[$|\mathcal{J}_u-\bar{\mathcal{J}}_{u}^{\mathrm{o}}|$ \emph{vs} $\lambda$\label{Fig:luciaJu}]{\includegraphics[scale=.4,trim=2cm 1.75cm 5cm 17cm,clip]{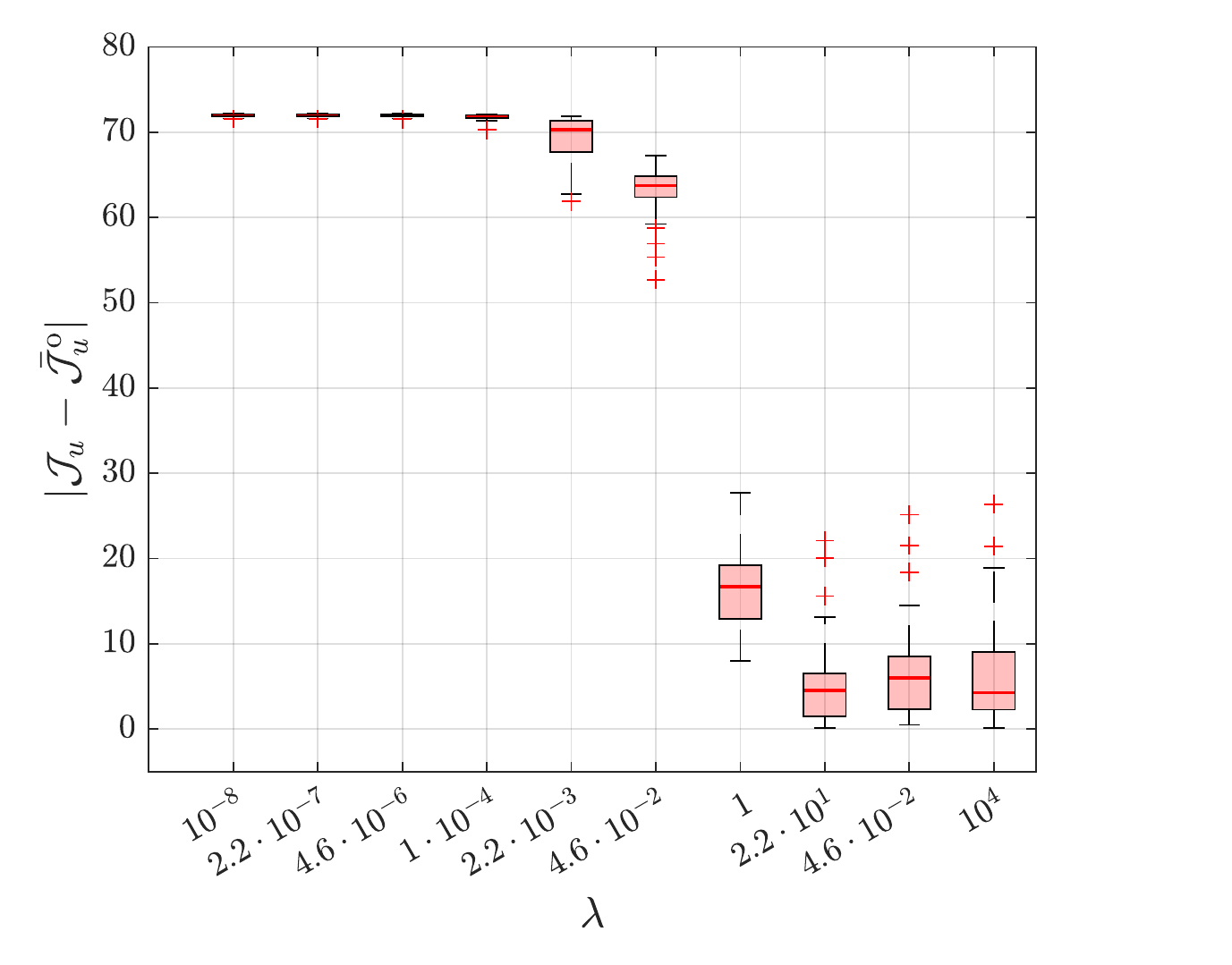}}
		\end{tabular}
		\caption{\rev{Closed-loop validation tests ($\overline{\mbox{SNR}}=18$~dB): absolute differences between the performance indexes attained with \cite{Felix2021} and the average values of those associated with the \emph{noisy} oracle MPC \emph{vs} penalties $\lambda$ over $30$ Monte Carlo predictors}.}\label{Fig:different_lucia}
	\end{figure}
		\begin{figure}[!tb]
		\centering
		\begin{tabular}{c}
			\subfigure[$|\mathcal{J}-\bar{\mathcal{J}}^{\mathrm{o}}|$ \emph{vs} $\bar{\lambda}_{\alpha}$ and $\lambda_{\sigma}$\label{Fig:berbJ}]{\includegraphics[scale=.4,trim=2cm 1.75cm 5cm 17cm,clip]{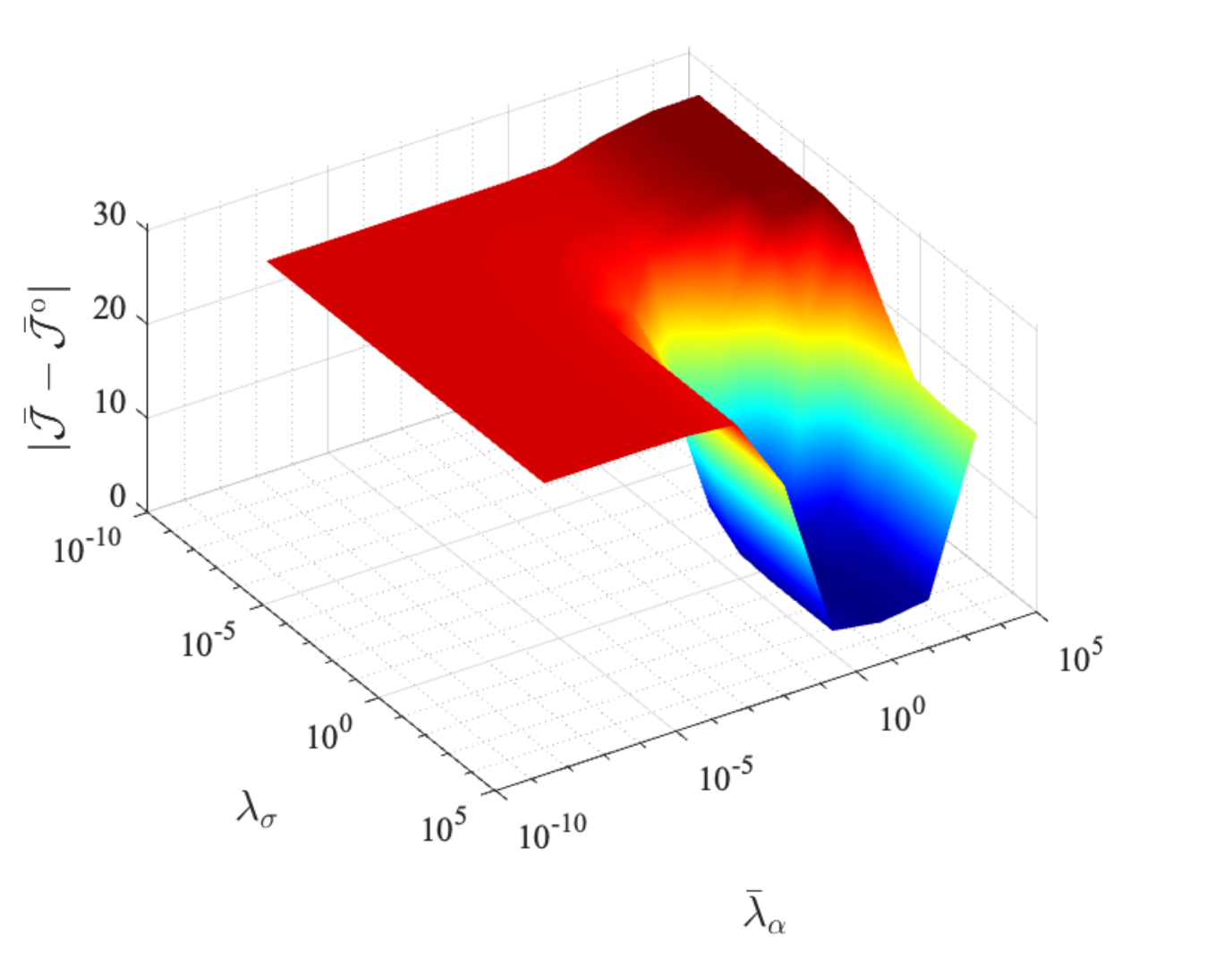}} \vspace{-.2cm}\\
			\subfigure[$|\mathcal{J}_u-\bar{\mathcal{J}}_{u}^{\mathrm{o}}|$ \emph{vs} $\bar{\lambda}_{\alpha}$ and $\lambda_{\sigma}$\label{Fig:berbJu}]{\includegraphics[scale=.4,trim=2cm 1.75cm 5cm 17cm,clip]{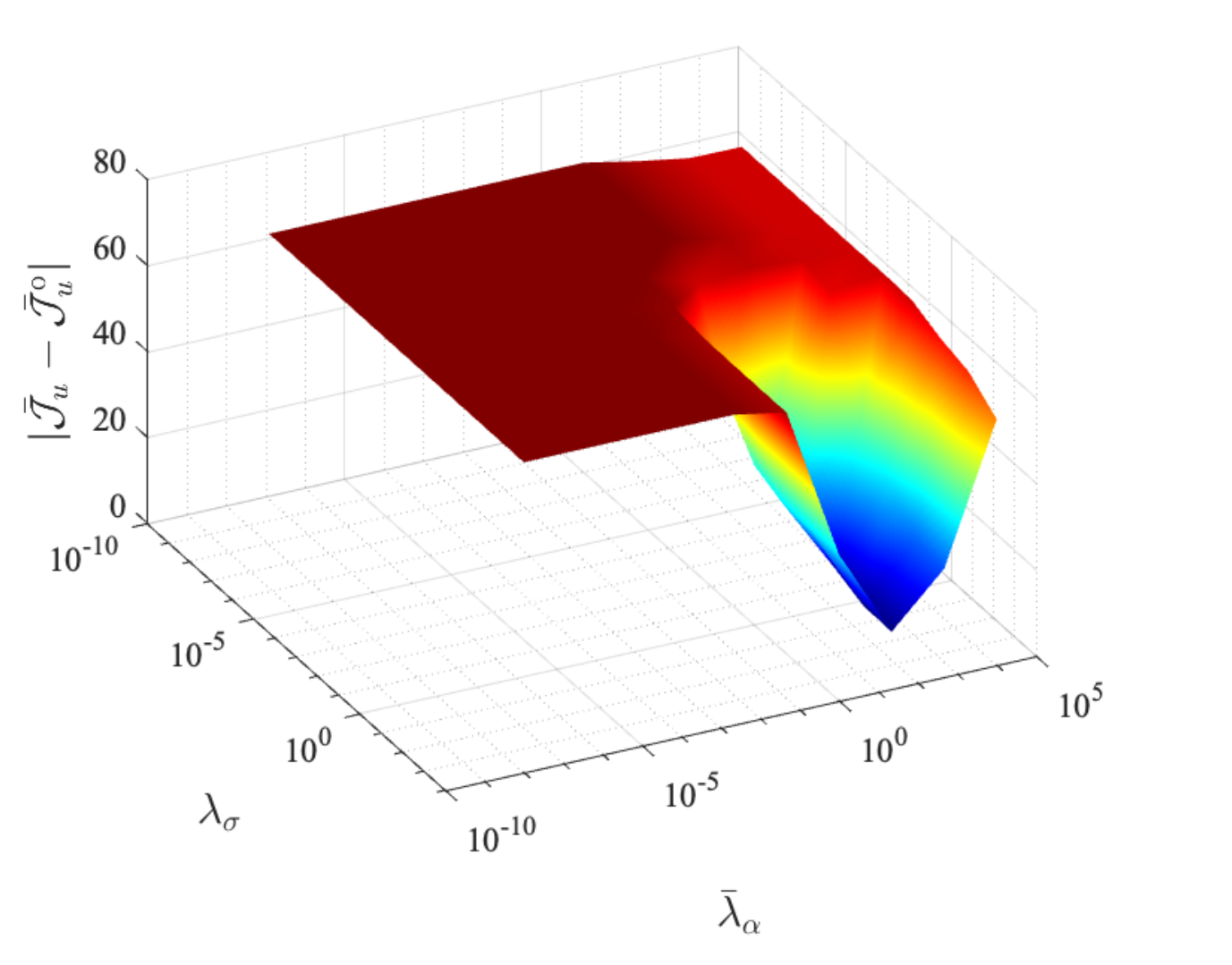}}
		\end{tabular}
		\caption{\rev{Closed-loop validation tests ($\overline{\mbox{SNR}}=18$~dB): absolute differences between the performance indexes attained with \cite{berberich2020data} and the average values of those associated with the \emph{noisy} oracle MPC \emph{vs} penalties $\bar{\lambda}_{\alpha}=\lambda_{\alpha}\bar{\varepsilon}$ and $\lambda_{\sigma}$ over $30$ Monte Carlo predictors}.}\label{Fig:different_berb}
	\end{figure}
\begin{figure}[!tb]
	\centering
	\begin{tabular}{c}
		\subfigure[$|\mathcal{J}-\bar{\mathcal{J}}^{\mathrm{o}}|$ \emph{vs} approaches\label{Fig:allJ}]{\includegraphics[scale=.4,trim=2.5cm 1.75cm 1cm 17cm,clip]{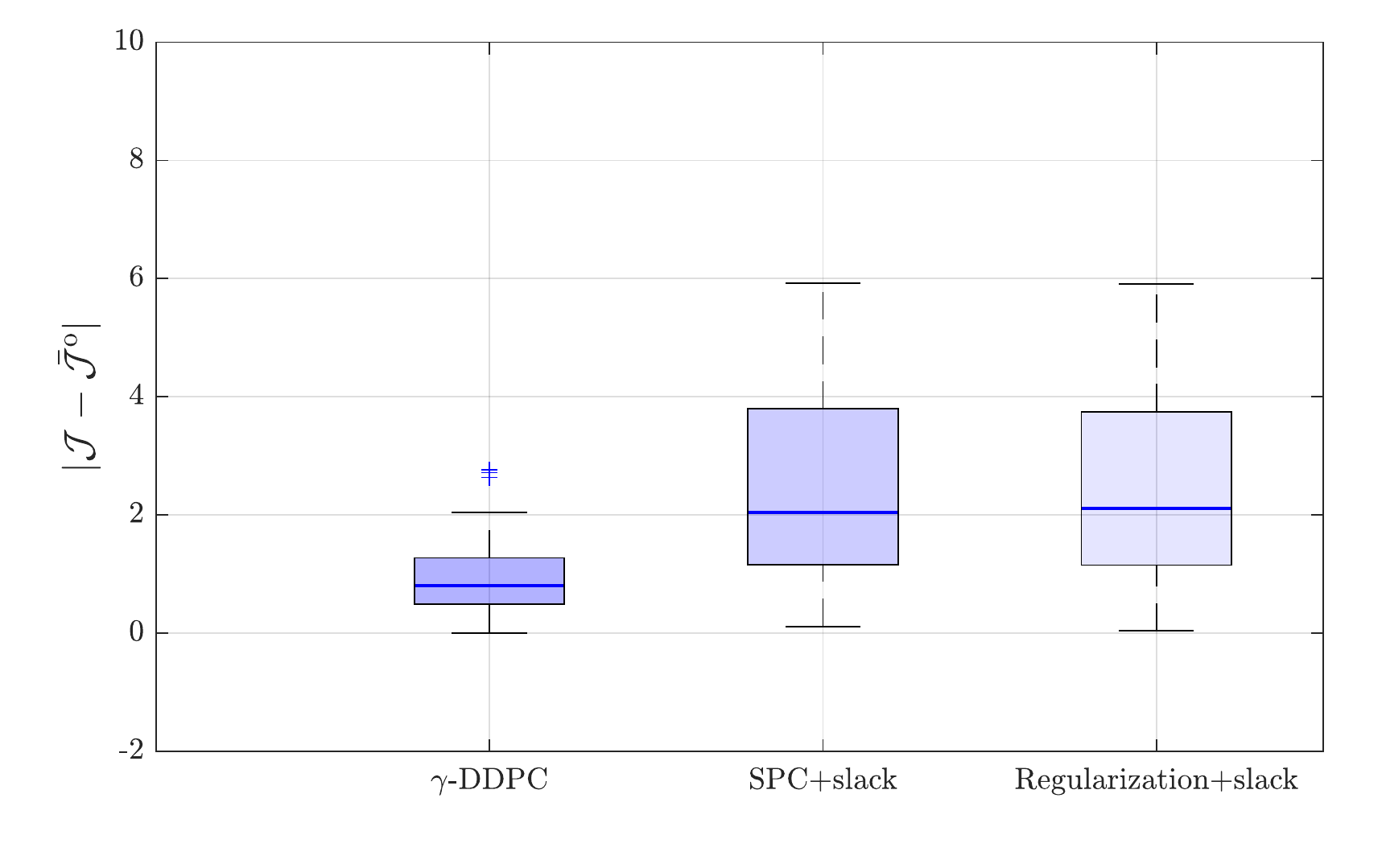}} \vspace{-.2cm}\\
		\subfigure[$|\mathcal{J}_u-\bar{\mathcal{J}}_{u}^{\mathrm{o}}|$ \emph{vs} approaches\label{Fig:allJu}]{\includegraphics[scale=.4,trim=2.5cm 1.75cm 1cm 17cm,clip]{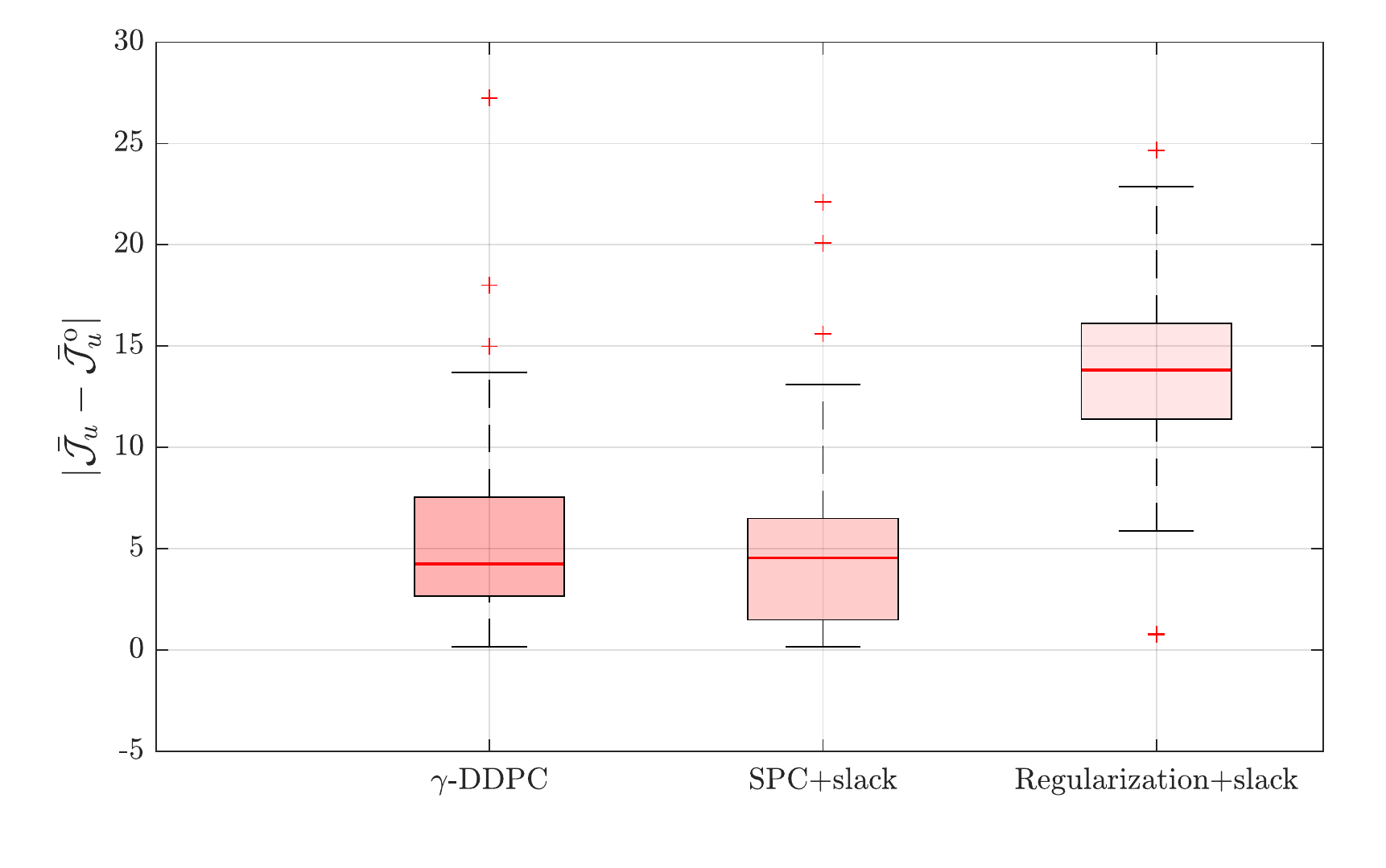}}
	\end{tabular}
	\caption{\rev{Closed-loop validation tests ($\overline{\mbox{SNR}}=18$~dB): absolute differences between the average performance indexes attained with he \emph{noisy} oracle MPC and $\mathcal{J}$ and $\mathcal{J}_{u}$ obtained with $\gamma$-DDPC, and the SPC+slack scheme in \eqref{eq:DD_Lucia} and the regularized approach with slack in \eqref{eq:DD_Berberich} with the best tuning over $30$ Monte Carlo predictors.}}\label{Fig:overall_notdorf}
\end{figure}
\begin{figure}[!tb]
	\centering
	\begin{tabular}{c}
		\subfigure[$|\mathcal{J}-\bar{\mathcal{J}}^{\mathrm{o}}|$ \emph{vs} $\lambda_{1}$ and $\lambda_{2}$\label{Fig:dorfJ}]{\includegraphics[scale=.4,trim=2cm 1.75cm 5cm 17cm,clip]{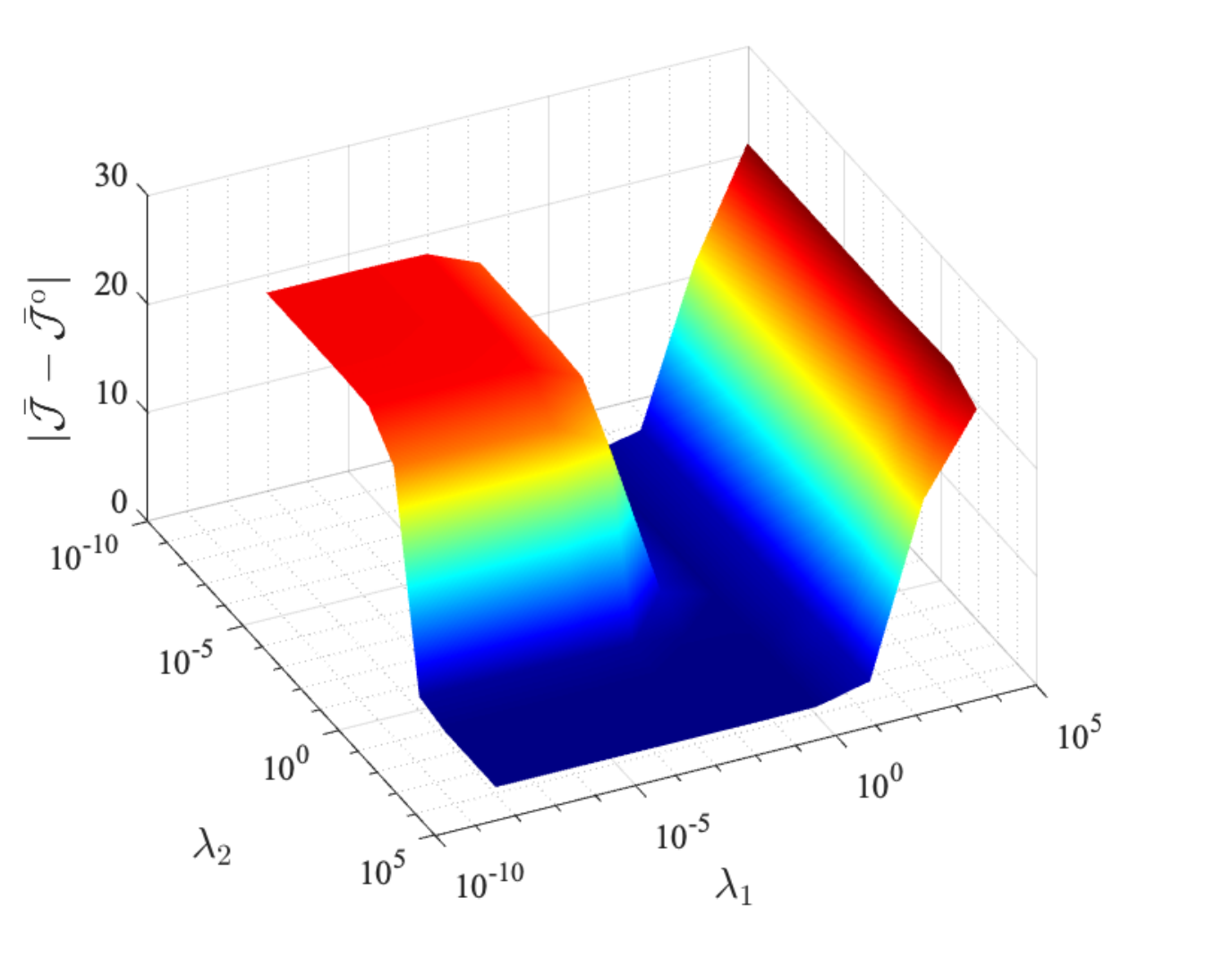}} \vspace{-.2cm}\\
		\subfigure[$|\mathcal{J}_u-\bar{\mathcal{J}}_{u}^{\mathrm{o}}|$ \emph{vs} ${\lambda}_{1}$ and $\lambda_{2}$\label{Fig:dorfJu}]{\includegraphics[scale=.4,trim=2cm 1.75cm 5cm 17cm,clip]{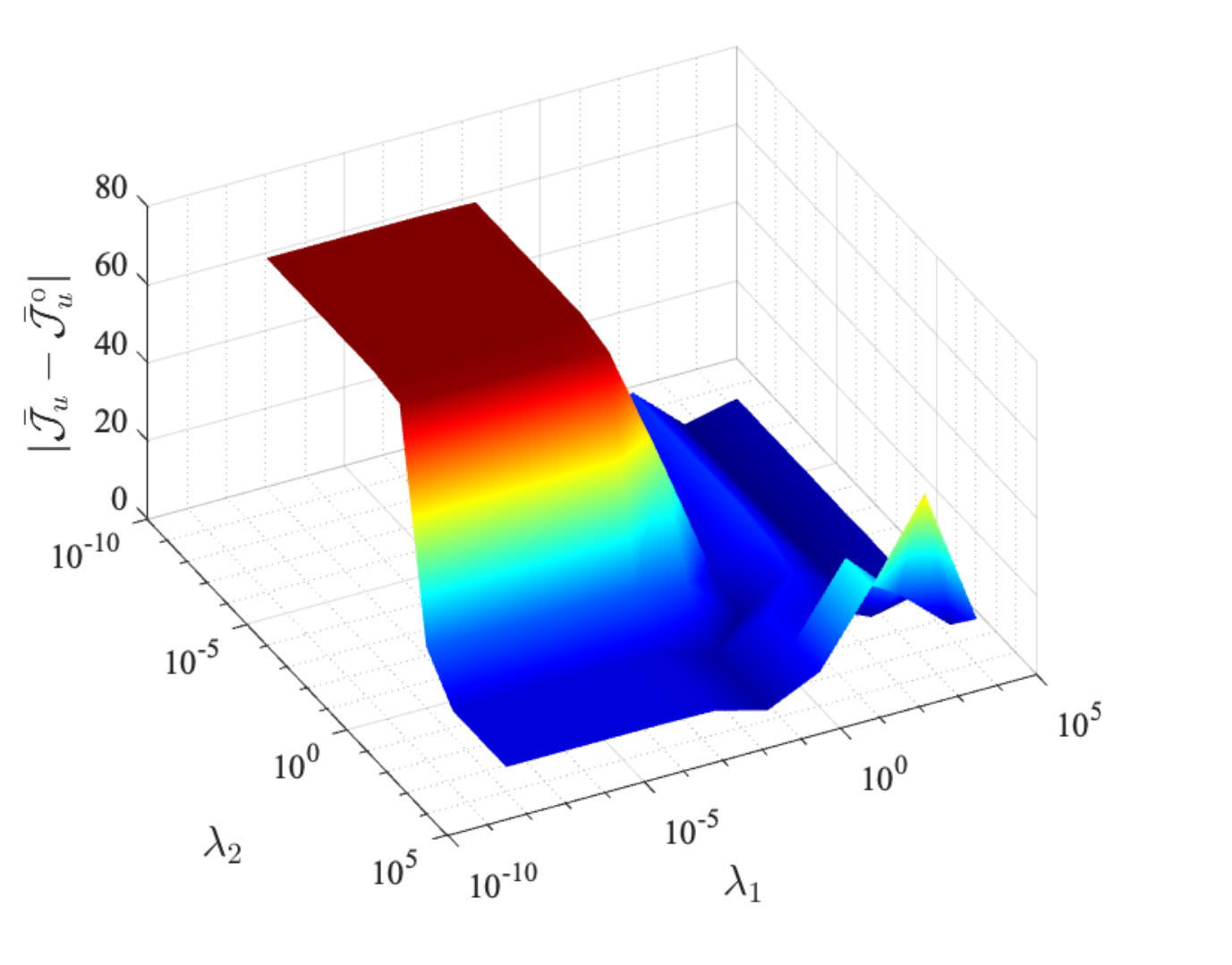}}
	\end{tabular}
	\caption{\rev{Closed-loop validation tests ($\overline{\mbox{SNR}}=18$~dB): absolute differences between the performance indexes attained with \cite{dorfler2022bridging} and the average values of those associated with the \emph{noisy} oracle MPC \emph{vs} penalties $\lambda_1$ and $\lambda_2$ over $5$ Monte Carlo predictors}.}\label{Fig:different_dorf}
\end{figure}
\begin{figure}[!tb]
	\centering
		\begin{tabular}{cc}
				\hspace*{-.4cm}\includegraphics[scale=.515,trim=2cm 2cm 10cm 21.7cm,clip]{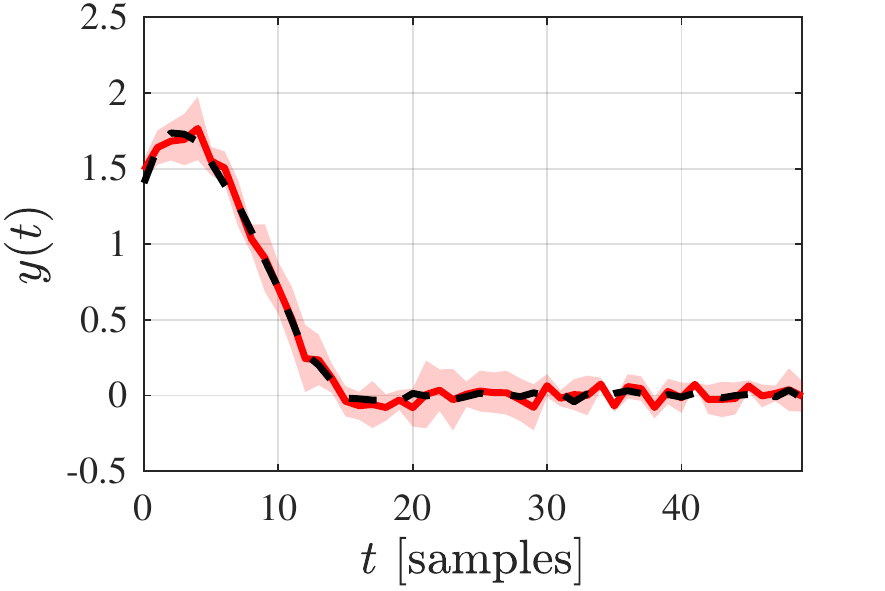} \hspace*{-.4cm}&\hspace*{-.4cm}	\includegraphics[scale=.515,trim=2cm 2cm 10cm 21.7cm,clip]{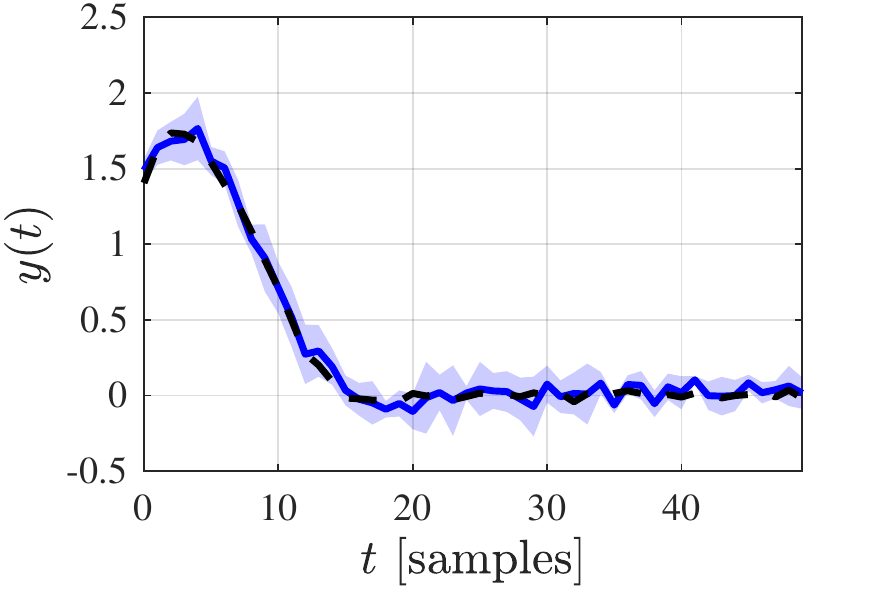}\vspace{-.2cm}\\
		\hspace*{-.4cm}\subfigure[$\lambda_{1}=10^{-8}$, $\lambda_{2}=10^{5}$]{\includegraphics[scale=.515,trim=2cm 2cm 10cm 21.7cm,clip]{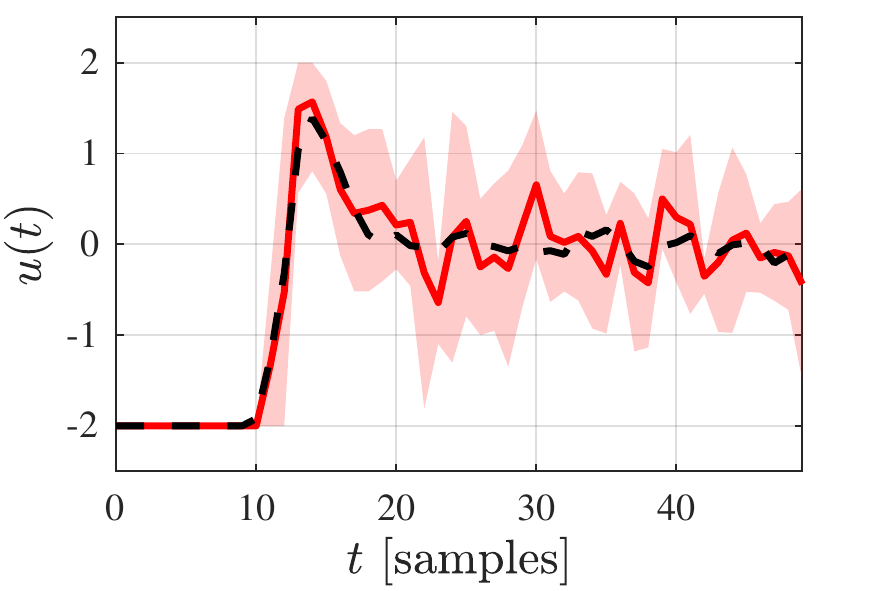}} \hspace*{-.4cm}&\hspace*{-.4cm}	\subfigure[$\lambda_{1}=1$, $\lambda_{2}=10^{-8}$]{\includegraphics[scale=.515,trim=2cm 2cm 10cm 21.7cm,clip]{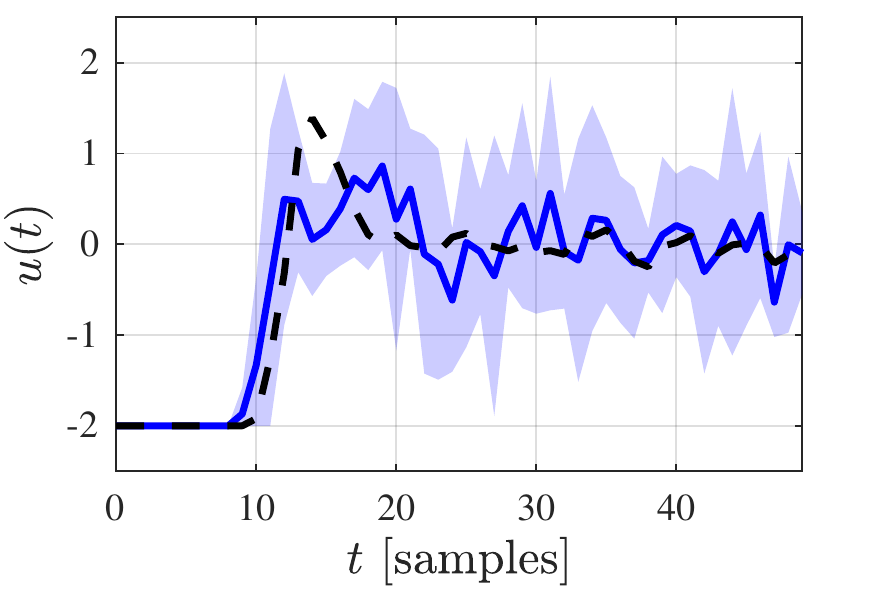}}
	\end{tabular}
	\caption{\rev{Average inputs for closed-loop validation tests ($\overline{\mbox{SNR}}=18$~dB): oracle (dashed black line) \emph{vs} average input (line) and standard deviation (shaded area) for different values of $\lambda_{1}$ and $\lambda_{2}$ in \eqref{eq:DD_Dorfler} over $5$ Monte Carlo predictors.}}\label{Fig:input_dorf}
\end{figure}
\rev{We now analyze the sensitivity of the three regularized DDPC approaches considered in Section~\ref{Sec:unifiednew} to different choices of their main tuning knobs, with the aim of experimentally validating the results stemming from the derived unified framework. The behavior of the performance indexes shown in \figurename{~\ref{Fig:different_lucia}} supports our conclusions. Indeed, the SPC+slack scheme proposed in \cite{Felix2021} tends to behave more closely to the oracle MPC for growing $\lambda$. Meanwhile, the input sequence fed to the system in closed loop tends to become equal to zero when $\lambda$ is small, concurrently causing a deterioration of the overall closed-loop performance.} As shown in \figurename{~\ref{Fig:different_berb}}, the choice of the regularization parameters is crucial to attain satisfactory performance \rev{when exploiting the approach proposed in \cite{berberich2020data}, balancing the need to have a meaningful control action and the one of rejecting noise. The attained} behavior validates the conclusions drawn in Section~\ref{Sec:unifiednew} \rev{with respect to the penalty $\lambda_{\sigma}$ in \eqref{eq:DD_Berberich}}. \rev{Indeed, higher values of this weight tends to improve the overall performance of the closed-loop. At the same time, since the dataset is finite and noisy, the results in \figurename{~\ref{Fig:different_berb}} highlight the importance of regularization for this DDPC formulation. Moreover,} these results show that regularizing the whole parameter vector $\alpha$\rev{, along with introducing a set of slacks,} requires a careful selection of \rev{both the associated} the regularization penalty. \rev{\rev{When compared with $\gamma$-DDPC, even with the best possible tuning, the schemes presented in \cite{Felix2021} and \cite{berberich2020data} result in the worst average performance with respect to the oracle MPC and a higher variability of the closed-loop behavior, see Figure \ref{Fig:allJ}. Note that, when the regularization penalty is properly tuned, the introduction of the slack variables in \eqref{eq:DD_Lucia} leads to an overall control effort similar to the oracle input sequence.}} \rev{Lastly, \figurename{~\ref{Fig:different_dorf}} and \figurename{~\ref{Fig:input_dorf}} corroborate the conclusions drawn in Section~\ref{Sec:unifiednew}. Indeed, it is clear that larger values of $\lambda_{2}$ and smaller $\lambda_{1}$ lead to performance that are comparable with that of the {$\gamma$-DDPC}. In particular, for}  $\lambda_{1}=10^{-8}$, it is clear that $\lambda_{2}$ has an effect similar to the one of $\eta$ and that (as expected) it is advisable to set $\lambda_{2}$ as large as possible. 

	\section{Conclusions}
	\rev{In this paper, {exploiting subspace identification tools,  we have provided an unifying framework for several  regularized data-driven predictive control schemes proposed in the literature, showing that they can be seen as variations of subspace predictive control. } This result extends the validity of these approaches beyond scenarios in which only measurement noise affects the system under control.} As a by-product, we have discussed the role of regularization, which is generally advocated in the literature as a tool to extend deterministic ideas to the noisy setting. \rev{By relying on the predictor decomposition proposed in the paper, we have further introduced the $\gamma$-DDPC problem, leading to a two-stage scheme where the effect of initial conditions, performance objectives and constraints is accounted for by solving two smaller optimization problems. By means of a numerical example, we show how the formulation at the core of $\gamma$-DDPC ease the interpretation of the effect of different regularization terms on the closed-loop behavior of the system, while validating the outcome of our discussions about the selection of the regularization penalties.}
	
	\rev{Future works will be devoted to the analysis of the closed loop properties of $\gamma$-DDPC, and to extend the latter to explicitly account for the error induced by the availability of a finite dataset. In addition, we will analyze the impact of regularization} when regularized DDPC schemes are applied to nonlinear systems.

	\bibliographystyle{plain}
	\bibliography{main}

\end{document}